\documentclass[12pt]{article}
\usepackage{amsmath,amsfonts,amssymb,amsthm}
\usepackage{graphicx,subfigure}
\usepackage{natbib}
\usepackage{hyperref}
\usepackage[margin=1in]{geometry}
\usepackage{setspace}
\usepackage{color}
\usepackage{multirow,array}
\usepackage{rotating}
\usepackage{eurosym}
\usepackage{tikz,mathtools}
\usepackage{todonotes}
\usepackage{amsthm}
\usepackage{soul}
\usepackage{diagbox}
\usepackage{graphicx}

\bibpunct[: ]{(}{)}{;}{a}{}{,}
\newtheorem{thm}{Theorem}
\newtheorem{pro}{Proposition}
\newtheorem{lem}{Lemma}
\newtheorem{alg}{Algorithm}
\newtheorem{ass}{Assumption}
\newtheorem{corr}{Corollary}
\newtheorem{ex}{Example}
\newtheorem{defin}{Definition}
\newtheorem{ob}{Observation}
\newtheorem{claim}{Claim}

\newenvironment{proposition}{\begin{pro}}{\end{pro}}

\newenvironment{corollary}{\begin{corr}}{\end{corr}}

\newenvironment{algorithm}{\begin{alg}}{\end{alg}}
\newtheorem{remark}{Remark}
\newcommand{\be}{\begin{equation}}
\newcommand{\ee}{\end{equation}}

\hypersetup{colorlinks=true,linkcolor=black,urlcolor=gray,citecolor=gray}
\begin{document}

\title{\textbf{Sequential Network Design\thanks{%
We are grateful to  the editor Pierpaolo Battigalli, an advisory editor, three anonymous referees,  Francis Bloch, Federico Echenique, Xiangliang Li, Xueheng Li, Eduardo Perez-Richet, Ludovic Renou, Sudipta Sarangi, Vasiliki Skreta, Satoru Takahashi, Tristan Tomala, Philip Ushchev, Nicolas Vieille,  and Yves Zenou for their valuable comments and insightful discussions. In addition, we appreciate the helpful feedback from participants at the Tenth Annual Conference on Network Science in Economics (Stanford University), GAMES 2024, the RUC Theory Workshop, the Tsinghua Network Workshop, CMiD 2022, the IHP PhD Seminar, Huazhong University of Science and Technology, HEC Paris, and the Chinese Academy of Sciences. All remaining errors are
 our own.}} }
\author{Yang Sun\thanks{%
Department of Economics, Southwestern University of Finance and Economics,
China. \textit{Email}: \href{mailto:sunyang789987@gmail.com}{%
sunyang789987@gmail.com}} 
\and 
Wei Zhao\thanks{School of Economics and Management, Tsinghua University, China. \textit{\ Email}: \href{mailto:wei.zhao@outlook.fr}{wei.zhao@outlook.fr } }
\and 
Junjie Zhou\thanks{%
School of Economics and Management, Tsinghua University, China. \textit{ Email}: \href{mailto:zhoujj03001@gmail.com}{zhoujj03001@gmail.com}} }
\date{November 24, 2025 }
\maketitle

\begin{abstract}
We examine dynamic network formation from a centralized perspective, where a
forward-looking social planner constructs one new link between previously
unconnected nodes in each period. The planner derives utility from the
discounted sum of benefits generated throughout the formation process.
Assuming the planner's instantaneous utility depends monotonically on the
aggregate number of walks of various lengths within the network, we derive
several key results. First, it is always optimal to form a nested split
graph at each stage, regardless of the discount function. Second, when the
planner is sufficiently myopic, the optimal strategy uniquely generates a
quasi-complete graph in each period. This finding provides a
micro-foundation for quasi-complete graphs as natural outcomes of greedy
network formation processes. Finally, we extend our analysis to weighted
networks, demonstrating the robustness of our results.

\end{abstract}

\newpage

\section{Introduction}

\label{sec:Intro}

Network economics examines how interaction topologies influence social
welfare. A central question in this field is the optimal allocation of
network links as resources. Unlike conventional scarce resources, network
links exhibit two distinctive characteristics. First, links between pairs of
nodes generate externalities that are far more complex than linear spillover
effects, making their allocation a systematic and intricate engineering
problem. Second, many types of networks cannot be formed instantaneously due
to capacity constraints, as network links require time to develop. For
example, constructing a road to connect isolated villages may take years,
while relationships need time and resources to mature before they can
effectively transmit peer effects or social capital. This temporal dimension
of network formation underpins empirical analyses of peer effects in social
networks, which often rely on studies of exogenous shocks for clear
identification or randomized controlled trials comparing different
intervention strategies. These studies generally assume that network
structures remain fixed within observation periods, an assumption supported
by empirical evidence. However, social planners must consider the welfare
generated throughout the entire network formation process, not just the
welfare associated with the final network structure. As a result, designing
the dynamic network formation path, rather than focusing solely on the
end-state, becomes particularly relevant in these contexts.

In this paper, we analyze a planner who sequentially allocates a
predetermined number of links among a fixed set of nodes by establishing a
single connection in each period between previously unconnected pairs. The
planner is forward-looking, deriving utility from the discounted sum of
benefits generated throughout the process. We employ a general class of
instantaneous preferences, requiring only that networks generating higher
aggregate numbers of walks of any length and larger spectral radius are
strictly preferred. These aggregate walk counts form the foundation for
walk-based centrality measures, a fundamental category of network
centralities. Katz-Bonacich centrality (\citealt{Katz1953} and %
\citealt{Bonacich1987}), eigenvector centrality (\citealt{Bonacich1972}),
community centrality (\citealt{Estrada2012}), and diffusion centrality (%
\citealt{Banerjee2013}) are all monotonically related to these statistics.
Beyond these instantaneous preferences, we accommodate any discount function
rather than restricting to geometric discounting. This general approach
encompasses several important special cases: static network design (%
\citealt{Belhaj2016} and \citealt{LI2023}) by weighting only final-period
utility, simultaneous multi-link allocation by assigning zero weight to
intermediate periods, and myopic planning by heavily discounting future
utilities.   The framework  naturally applies to contexts in which planners design and expand network infrastructure under capacity constraints. For instance, building road systems to connect remote villages or constructing power grid networks to link communities often requires years of planning and implementation due to limited construction resources and labor capacity. In such settings, the planner is concerned not only with the final network configuration but also  with the welfare generated during the formation process, as intermediate network structures can already yield  value.

Two key network classes emerge in our analysis. A nested split graph (NSG)
is characterized by the property that for any pair of nodes $i$ and $j$,
either the neighbors of $i$ form a subset of the neighbors of $j$, or vice
versa. A quasi-complete (QC) graph represents the largest possible clique
given the total number of links, with any remaining links connecting clique
members to a single outside node. QC graphs form a specific subclass of
NSGs. While multiple types of NSGs can exist for a fixed number of links, QC
graphs are unique up to isomorphism.

Our main results address two complementary aspects of network formation.
First, we prove that forming a nested split graph at each period is optimal,
regardless of the discount function (Theorem \ref{thm:unweighted-NSG}).
Second, we establish that when the planner is myopic, the optimal solution,
consistent with a greedy algorithm, results in a quasi-complete graph being
formed in each period (Theorem \ref{thm:unweighted-QC}).

The proof of our first main result follows a two-step approach. First, we
demonstrate (Lemma \ref{lem:wel-increase}) that when neither node's
neighborhood is a subset of the other's, reallocating one node's distinct
neighbors to the other strictly increases both the aggregate number of walks
of length greater than two and the spectral radius of the network. This
result generalizes the key insight from \cite{Belhaj2016}, which showed that
such operations enhance aggregate Katz-Bonacich centralities and their
squared values. In the second step, we prove that for any network formation
path producing a non-NSG in any period, we can construct an alternative
feasible path (Algorithm \ref{alg:dynamic-reallocation}) that generates NSGs
in every period and strictly dominates the original path.

For our second main result, we leverage the observation that adding a single
link to a QC graph produces at most two types of NSGs, one of which is
another QC graph. We then establish that QC graphs dominate the alternative
NSG type in terms of both aggregate walks of all lengths and spectral radius
(Lemma \ref{lem:weighted-KB2} (ii)). A significant implication of this lemma
is the refinement of theoretical predictions in static network design by
excluding certain subclasses of NSGs from optimality.

Next, we extend our analysis to scenarios where the planner sequentially
allocates one unit of weight per period, with each link's weight bounded by
one.

First, we establish that when the planner aims to maximize the discounted
sum of Katz-Bonacich (KB) centralities, the optimal weighted network
formation path produces unweighted networks in each period (Proposition \ref%
{pro:weighted-KB2-KB}). This result indicates that the flexibility of weight
allocation offers no advantage over discrete link allocation in this
context. The proof employs a key lemma, which demonstrates that the sum of
KB centralities is convex in the underlying network. We complete the proof
by showing that the set of feasible network formation paths forms a convex
set, with paths comprising unweighted networks as extreme points.

Second, we prove that if instantaneous utility increases with the sum of
squared KB centralities, the optimal formation path produces a weighted NSG
in each period, regardless of the discount function, and an unweighted QC
graph in each period when the planner is sufficiently myopic (Proposition %
\ref{pro:weighted-KB2}). To establish this, we apply another lemma from the
same paper that provides sufficient conditions for weight reallocation to
improve the sum of convex functions of KB centralities. For the first part,
we explicitly construct a dominant perturbation for any weighted network
formation path containing at least one non-weighted NSG. For the second
part, we show that when allocating weight to a QC graph, optimal solutions
lie within a restricted subclass of weighted NSGs that represent convex
combinations of two unweighted NSGs. Finally, we prove that the QC graph
dominates any such convex combination in terms of aggregate walks of all
lengths.

\subsection*{Literature Review}

Network formation, alongside network centrality, is a foundational topic in
network economics. The literature on network formation is divided into two
major strands. The first strand focuses on networks formed through \textit{%
decentralized processes}, which further branches into two categories. The
first branch examines static network formation with decentralized concerns,
including seminal works by \cite{Jackson1996}, \cite{Bala2000}, \cite%
{Bloch2006}, \cite{Galeotti2010}, \cite{Cabrales2011} and \cite%
{Christakis2020}. The second branch investigates dynamic network formation
with decentralized concerns, featuring contributions from \cite{Bala2000}, 
\cite{Watts2001}, \cite{Jackson2002a}, \cite{Dutta2005}, \cite{Page2005}, 
\cite{Koenig2014} and \cite{Song2020}.\footnote{%
Note that some of the literature on decentralized network formation also
discusses efficiency. However, the efficient network in these setups is
typically easy to characterize, and the primary objective is to derive
sufficient conditions for the efficiency of stable networks formed through
decentralized processes.} The second strand, where networks are formed based
on \textit{centralized concerns}, is also known as network design. In this
strand, the existing literature primarily focuses on static network design,
including works by \cite{Baetz2015}, \cite{Belhaj2016}, \cite{Hiller2017}, \cite{LI2023}, and others. To our knowledge, our paper addresses a previously
unexplored area: dynamic network formation with centralized concerns.%
\footnote{%
Some papers in the branch of dynamic network formation with decentralized
concerns also address the efficiency of long-run stable or stationary
networks. However, the efficiency benchmark they adopt is usually the static
efficient network.}

Unlike approaches that focus solely on the final network structure, our
paper analyzes the entire network formation path. The broad relationship
between our paper and the existing literature is also summarized in the
following table.

\begin{table}[!ht]
\centering
\resizebox{1\textwidth}{!}{\begin{tabular}{|l|l|l|}
\hline
\diagbox[width=11em]{Strand}{Branch} & \textbf{Static} & \textbf{Dynamic} \\ \hline
\textbf{Decentralized} &
  \begin{tabular}[c]{@{}l@{}} 
  \cite{Jackson1996}, \cite{Bala2000}, \\ 
  \cite{Bloch2006}, \cite{Galeotti2010}, \\ 
  \cite{Cabrales2011}
  \end{tabular} &
  \begin{tabular}[c]{@{}l@{}} 
  \cite{Bala2000}, \cite{Watts2001}, \\ 
  \cite{Jackson2002a}, \cite{Dutta2005}, \\ 
  \cite{Page2005}, \cite{Koenig2014}, \\ 
  \cite{Song2020}
  \end{tabular} \\ \hline
\textbf{Centralized} &
  \begin{tabular}[c]{@{}l@{}} 
  \cite{Belhaj2016}, \cite{Baetz2015}, \\ 
\cite{Hiller2017},\cite{LI2023}
  \end{tabular} &
  \textbf{Our Paper} \\ \hline
\end{tabular}}
\caption{Network Formation Literature}
\label{tab:net-form-literature}
\end{table}

Another distinguishing feature of this paper is the instantaneous preference
of the planner over the network. The literature can be classified into two
categories based on how the criterion of efficiency depends on the network
structure. The first category includes works where nodes benefit from direct
connections, such as \cite{Jackson1996},\cite{Dutta1997}
, \cite{Bala2000},\cite{Jackson2002}, \cite{Watts2001}, \cite{Dutta2005}, \cite{Bloch2006}, \cite{Song2020} and \cite{BRAVARD2025}. The second category considers cases
where agents' payoffs are endogenously determined through equilibrium in
network games. These include models of strategic substitution (%
\citealt{Galeotti2010}, \citealt{Billand2015}, \citealt{Leeuwen2019}) and
strategic complementarity (\citealt{Cabrales2011}, \citealt{Baetz2015},  \citealt{Belhaj2016}, \citealt{Hiller2017}, \citealt{LI2023}).

In our paper, the planner's preference depends on the network topology
through a key network statistic -- namely, the weighted sum of the aggregate
number of walks of various lengths. This criterion generalizes the
literature in the second category, which typically adopts the linear
quadratic network game introduced by the seminal work of \cite{Ballester2006}%
. Moreover, our paper also encompasses scenarios where the planner aims to
maximize other walk-based centralities, such as diffusion centrality (%
\citealt{Banerjee2013}, \citealt{Cruz2017}, \citealt{Banerjee2018}),
spectral radius (\citealt{Brualdi1985}), and the sum of aggregate walks of
length two (\citealt{Abrego2009}).

The remainder of this paper is organized as follows. Section 2 introduces
the formal setup and discusses the planner's problem. Section 3 presents our
main theoretical results. Section 4 extends these findings to weighted
networks. Section 5 concludes with directions for future research. All
proofs are provided in the Appendix.

\section{The Model}

A network consisting of a set $N=\{1,\ldots,n\}$ of nodes is represented by an adjacency matrix $\mathbf{G}=(g_{ij})_{n\times n}$, where $g_{ij}=g_{ji}=1$ if nodes $i$ and $j$ are linked, and $g_{ij}=g_{ji}=0$ otherwise.\footnote{See Section \ref{sec:weighted} for a detailed discussion of weighted networks.} Let $\mathbf{E}_{ij}$ denote the matrix with $1$ at the $(i,j)$ and $(j,i)$ entries, and $0$ at all other entries. We say that a network $\mathbf{\hat{G}}$ \textit{succeeds} network $\mathbf{G}$ if $\mathbf{\hat{G}}$ can be obtained by adding a new link to $\mathbf{G}$, i.e., if there exist two nodes $i,j$ such that $g_{ij}=0$ and $\mathbf{\hat{G}}=\mathbf{G}+\mathbf{E}_{ij}$. Let $\mathbb{S}(\mathbf{G})$ denote the set of networks that succeed $\mathbf{G}$.

Starting with an empty network, a planner dynamically constructs the network over $T$ periods by adding one new link in each period. Formally, the planner's strategy is a sequence of networks 
\begin{equation*}
\mathbf{s}=\left(\mathbf{G}\left(1\right),\ldots,\mathbf{G}\left(T\right)\right)
\end{equation*}
such that $\mathbf{G}\left(t\right)\in\mathbb{S}\left(\mathbf{G}\left(t-1\right)\right)$ for any $t=1,\ldots,T$ (with $\mathbf{G}\left(0\right)=\mathbf{0}$). Let $S$ be the set of all such sequences of networks that satisfy feasibility. Since the networks are unweighted, the set $S$ is finite and $T\leq n(n-1)/2$.

\subsection{The Optimization Problem}

\label{sub:perference} The planner cares about the entire stream of networks. Given a sequence of networks $\mathbf{s}=(\mathbf{G}(t))_{t=1}^{T}\in S$, the planner evaluates this sequence according to the value function
\begin{equation*}
v(\mathbf{s}):=\sum_{t=1}^{T}D(t)u(\mathbf{G}(t)),
\end{equation*}
where $u\left(\mathbf{G}\left(t\right)\right)$ represents the instantaneous utility in period $t$ (obtained from the network $\mathbf{G}\left(t\right)$) and $D\left(t\right)\geq 0$ is the discount factor for period $t$. The value function $v(\mathbf{s})$ aggregates the discounted utilities across all periods, capturing the planner's intertemporal preferences over the network formation process.

The planner seeks to select a sequence of networks that maximizes the value function. Formally, the planner's problem is described as: 
\begin{equation}
\max_{\mathbf{s}\in S}v(\mathbf{s}).  \label{design problem}
\end{equation}
The solution to Problem \eqref{design problem} always exists because the set $S$ is finite.

We are flexible in the choice of discount factors $\left(D(t)\right)_{t=1}^{T}$, requiring only that they be nonnegative. A commonly used case is geometric discounting: $D(t)=\delta^{t}$ for some $\delta>0$. Note that this $D(t)$ is not necessarily required to decrease in $t$, as $\delta$ may exceed one.

Two limiting cases of geometric discounting are particularly interesting. On one extreme, as $\delta\rightarrow+\infty$, the normalized geometric discount factors $D(t)=\frac{\delta^{t}}{\delta^{T}}$ converge to the case of farsightedness:

\begin{defin}
\label{def-FS} The planner is farsighted if $D\left(t\right)=\begin{cases} 0 & \text{if } 1\leq t\leq T-1 \\ 1 & \text{if } t=T \end{cases}$. \label{eq:farsighted}
\end{defin}

A farsighted planner prefers $\mathbf{s}$ to $\mathbf{\hat{s}}$ if the final network $\mathbf{G}\left(T\right)$ under $\mathbf{s}$ is better than $\mathbf{\hat{G}}\left(T\right)$ under $\mathbf{\hat{s}}$. The intermediate networks $\mathbf{G}(t)$ for $t<T$ do not affect a farsighted planner's utility. The limit of geometric discounting as $\delta\rightarrow+\infty$ reflects a planner who focuses solely on achieving a desirable final network.

On the other extreme, when the geometric discounting factor $\delta\rightarrow 0^{+}$, preferences converge to myopia:

\begin{defin}
\label{def-mypoic} The planner is myopic if, for any $\mathbf{s}=\left(\mathbf{G}\left(1\right),\ldots,\mathbf{G}\left(T\right)\right)$ and $\mathbf{\hat{s}}=\left(\mathbf{\hat{G}}\left(1\right),\ldots,\mathbf{\hat{G}}\left(T\right)\right)$ with $t'$ being the first time when $u(\mathbf{G}(t'))\neq u(\hat{\mathbf{G}}(t'))$, 
\begin{equation*}
u(\mathbf{G}(t'))>u(\hat{\mathbf{G}}(t'))\implies v(\mathbf{s})>v(\hat{\mathbf{s}}).
\end{equation*}
\end{defin}

A myopic planner focuses on the immediate structure of the network and heavily discounts future payoffs. Equivalently, there exists a cutoff $\epsilon>0$ such that the planner is myopic if and only if the geometric discount factor $\delta$ is smaller than $\epsilon$.

\begin{remark}
By choosing an appropriate discount function, the model can also encompass the case where the planner is capable of establishing multiple links within a single period. For instance, when $D\left(t\right)=0$ and $D\left(t+1\right)>0$, it is as if the planner constructs two links at once. An extreme case is described in Definition \ref{def-FS}, where the planner chooses all $T$ links at once.
\end{remark}

Next, we impose the following assumption on the planner's instantaneous utility function $u(\cdot)$. Given a network $\mathbf{G}$, for any nonnegative integer $k$, denote 
\begin{equation*}
W^{k}(\mathbf{G}):=\mathbf{1}'\mathbf{G}^{k}\mathbf{1},
\end{equation*}
where $\mathbf{1}$ denotes the $n$-dimensional vector of $1$s. As is well-known in the network literature, $W^{k}(\mathbf{G})$ counts the total number of walks of length $k$ in the network. Let $\lambda_{\max}(\mathbf{G})$ denote the spectral radius of the network. We write $\mathbf{G}\cong\mathbf{\hat{G}}$ when two networks $\mathbf{G}$ and $\mathbf{\hat{G}}$ are isomorphic.

\begin{ass}
\label{ass:preference} For any two networks $\mathbf{G}$ and $\mathbf{\hat{G}}$ with the same total number of links, if 
\begin{equation*}
\begin{cases} W^{k}(\mathbf{G})>W^{k}(\mathbf{\hat{G}})\text{, for any }k\geq 2 \\ \lambda_{\max}(\mathbf{G})>\lambda_{\max}(\mathbf{\hat{G}}) \end{cases}
\end{equation*}
then $u(\mathbf{G})>u(\mathbf{\hat{G}})$. Moreover, if $\mathbf{G}\cong\mathbf{\hat{G}}$, then $u(\mathbf{G})=u(\mathbf{\hat{G}})$.
\end{ass}

Assumption \ref{ass:preference} is very mild, as it imposes restrictions only on pairs of networks where one uniformly and strictly dominates the other in terms of aggregate walks of not only arbitrary length but also in the limit as length approaches infinity (noting that $\lambda_{\max}(\mathbf{G})>\lambda_{\max}(\mathbf{\hat{G}})$ is equivalent to $\lim_{k\rightarrow\infty}\frac{W^{k}(\mathbf{G})}{W^{k}(\mathbf{\hat{G}})}=\infty$). For network pairs not comparable under this strict dominance relation, the assumption allows flexibility in preferences. Specifically, when $\lambda_{\max}(\mathbf{G})=\lambda_{\max}(\mathbf{\hat{G}})$, or there exist two positive integers $l$ and $l'$ such that $W^{l}(\mathbf{G})>W^{l}(\mathbf{\hat{G}})$ while $W^{l'}(\mathbf{G})\leq W^{l'}(\mathbf{\hat{G}})$, then the planner is able to prefer either network without violating the assumption.

\subsection{Examples and Applications}
\label{sub:examples}
Assumption \ref{ass:preference} requires the planner's preference to respect network rankings according to aggregate walks of arbitrary length. Instantaneous utilities of the form $u(\mathbf{G})=\sum_{k\geq 0}\rho_{k}\mathbf{1}'\mathbf{G}^{k}\mathbf{1}$, where $\rho_{k}>0$ and the series converges, always satisfy Assumption \ref{ass:preference}.\footnote{A sufficient condition for convergence on all networks with $n$ nodes is that the power series $\sum_{k\geq 0}\rho_{k}x^{k}$ has radius of convergence at least $n-1$, the maximum possible spectral radius of any graph.} We highlight several widely used specifications.

\begin{enumerate}
\item Assume $0<\phi<\frac{1}{\lambda_{\max}(\mathbf{G})}$. For any integer $\alpha\geq 0$, define 
\begin{equation}  \label{eq:KB&walks}
b(\alpha,\phi,\mathbf{G}):=\mathbf{1}'(\mathbf{I}-\phi\mathbf{G})^{-\alpha}\mathbf{1}=\sum_{k=0}^{\infty}\binom{\alpha+k-1}{k}\phi^{k}\mathbf{1}'\mathbf{G}^{k}\mathbf{1},
\end{equation}
where $\binom{\alpha+k-1}{k}$ is the binomial coefficient. The utility function $u(\mathbf{G})=b(\alpha,\phi,\mathbf{G})$ satisfies Assumption \ref{ass:preference} for any integer $\alpha\geq 0$.

The cases $\alpha=1$ and $\alpha=2$ are  of particular significance. \cite{Ballester2006} show that in a linear quadratic network game, the unique equilibrium profile is $\mathbf{x}^{\ast}=(\mathbf{I}-\phi\mathbf{G})^{-1}\mathbf{1}$, yielding aggregate effort $\sum_{i=1}^{n}x_{i}^{\ast}=b(1,\phi,\mathbf{G})$ and aggregate welfare $\frac{1}{2}\sum_{i=1}^{n}(x_{i}^{\ast})^{2}=\frac{1}{2}b(2,\phi,\mathbf{G})$. Thus, a planner maximizing either aggregate effort or total welfare falls within our framework.

\item For $\beta>0$, define 
\begin{equation}
\label{eq:C&walks}
c(\beta,\mathbf{G}):=\mathbf{1}'e^{\beta\mathbf{G}}\mathbf{1}=\sum_{k=0}^{\infty}\frac{\beta^{k}}{k!}\mathbf{1}'\mathbf{G}^{k}\mathbf{1}
\end{equation}
as the total network communicability \citep{Estrada2012,Benzi2013}. This measure quantifies the overall effectiveness of communication across the network by summing the communicability between all pairs of nodes. For any $\beta>0$, the utility $u(\mathbf{G})=c(\beta,\mathbf{G})$ satisfies Assumption \ref{ass:preference} since the exponential series converges and assigns positive weight $\rho_{k}=\frac{\beta^{k}}{k!}>0$ to walks of length $k$. As \cite{Benzi2013} demonstrate, this measure is particularly valuable for comparing communication efficiency across network configurations.

\item The utility function $u(\mathbf{G})=\lambda_{\max}(\mathbf{G})$, or any monotonic transformation thereof, also satisfies Assumption \ref{ass:preference}, though it is not explicitly walk-based.\footnote{Note that $W^{k}(\mathbf{G})>W^{k}(\mathbf{\hat{G}})$ for all $k\geq 2$ implies $b(\alpha,\phi,\mathbf{G})>b(\alpha,\phi,\mathbf{\hat{G}})$ and $c(\beta,\mathbf{G})>c(\beta,\mathbf{\hat{G}})$, but only $\lambda_{\max}(\mathbf{G})\geq\lambda_{\max}(\mathbf{\hat{G}})$. Equality can occur when $\mathbf{G}$ has isolated components with the largest being isomorphic to $\mathbf{\hat{G}}$. Assumption \ref{ass:preference} accommodates such cases by permitting the planner to be indifferent between such networks.} The maximal spectral radius problem thus emerges as a special case of problem \eqref{design problem} when the planner benefits from the network's spectral radius.\footnote{Finding the graph with maximum spectral radius for a given number of links was posed by \cite{Brualdi1985} and remains open after 35 years; see \cite{Radanovic2024} for recent progress.}
\end{enumerate}
Beyond centrality measures, our framework also extends to other network game models in which the Katz–Bonacich centrality, or its variants and transformations, determine equilibrium outcomes.  Let $\mathbf{b}(\mathbf{G},\delta)=(\mathbf{I}-\delta\mathbf{G})^{-1}\mathbf{1}$ denote the vector of Katz-Bonacich centrality and $b(\mathbf{G},\delta)=\mathbf{1}'(\mathbf{I}-\delta\mathbf{G})^{-1}\mathbf{1}$ its sum. We illustrate with three examples:

\begin{enumerate}
\item[(i)] (\textbf{Global Substitution}) Consider the linear quadratic network game of \cite{Ballester2006}, where player $i$'s utility is
\begin{equation*}
u_{i}(x_{i},\mathbf{x}_{-i})=x_{i}-\frac{1}{2}x_{i}^{2}-\phi\sum_{k\neq i}x_{i}x_{k}+\delta\sum_{k=1}^{n}g_{ik}x_{i}x_{k}.
\end{equation*}
The term $\phi\sum_{k\neq i}x_{i}x_{k}$ captures global interaction effects corresponding to strategic substitutability across all players, with $\phi\geq 0$ measuring the intensity of this interdependence. The equilibrium aggregate effort  is $\mathbf{1}'\mathbf{x}^{\ast}=\frac{b(\mathbf{G},\frac{\delta}{1-\phi})}{1+\phi b(\mathbf{G},\frac{\delta}{1-\phi})}$, which is monotone in $b(\mathbf{G},\frac{\delta}{1-\phi})$ and thus satisfies Assumption \ref{ass:preference}.

\item[(ii)]  (\textbf{Multiple Activities})
\cite{D2x} analyze a network model where each player $i$ chooses levels of two activities $(x_{i}^{A},x_{i}^{B})=\mathbf{x}_{i}$ with utility
\begin{equation*}
u_{i}(\mathbf{x}_{i},\mathbf{x}_{-i})=x_{i}^{A}+x_{i}^{B}-\left\{\frac{1}{2}(x_{i}^{A})^{2}+\frac{1}{2}(x_{i}^{B})^{2}+\beta x_{i}^{A}x_{i}^{B}\right\}+\delta\sum_{j}g_{ij}x_{i}^{A}x_{j}^{A}+\delta\sum_{j}g_{ij}x_{i}^{B}x_{j}^{B},
\end{equation*}
where $\frac{1}{2}(x_{i}^{A})^{2}+\frac{1}{2}(x_{i}^{B})^{2}+\beta x_{i}^{A}x_{i}^{B}$ represents the cost of actions and $\delta\sum_{j}g_{ij}x_{i}^{A}x_{j}^{A}+\delta\sum_{j}g_{ij}x_{i}^{B}x_{j}^{B}$ captures network externalities. The aggregate equilibrium activities are
\begin{equation*}
\mathbf{1}'\mathbf{x}^{A}=\sum_{t=0}^{\infty}\left(\frac{\delta^{t}}{2(1+\beta)^{t+1}}+\frac{\delta^{t}}{2(1-\beta)^{t+1}}\right)\mathbf{1}'\mathbf{G}^{t}\mathbf{1}
\end{equation*}
and
\begin{equation*}
\mathbf{1}'\mathbf{x}^{B}=\sum_{t=0}^{\infty}\left(\frac{\delta^{t}}{2(1+\beta)^{t+1}}-\frac{\delta^{t}}{2(1-\beta)^{t+1}}\right)\mathbf{1}'\mathbf{G}^{t}\mathbf{1}.
\end{equation*}
Maximizing aggregate activity  $A$ or $B$ satisfies Assumption \ref{ass:preference} whenever $\beta<0$.

\item[(iii)](\textbf{Congestion Effects})
\cite{Currarini2017} study a network game with congestion effects between distance-two neighbors, where player $i$'s payoff is
\begin{equation*}
u_{i}(x_{i},\mathbf{x}_{-i})=x_{i}-\frac{1}{2}x_{i}^{2}+\delta\sum_{k=1}^{n}g_{ik}x_{i}x_{k}-\gamma\sum_{k=1}^{n}g_{ik}^{[2]}x_{i}x_{k},
\end{equation*}
where $g_{ik}^{[2]}$ is the $ik$-th element of $\mathbf{G}^{2}$ and the term $-\gamma\sum_{k=1}^{n}g_{ik}^{[2]}x_{i}x_{k}$ captures strategic substitution between players at distance two in the network. The first-best strategy profile can be written as a linear combination of two Katz-Bonacich centralities:
\begin{equation*}
\mathbf{x}^{\ast}=\frac{\beta_{1}}{\beta_{1}-\beta_{2}}\mathbf{b}(\mathbf{G},\beta_{1})-\frac{\beta_{2}}{\beta_{1}-\beta_{2}}\mathbf{b}(\mathbf{G},\beta_{2}),
\end{equation*}
where $\beta_{1}=\frac{\delta+\sqrt{\delta^{2}-4\gamma}}{2}$ and $\beta_{2}=\frac{\delta-\sqrt{\delta^{2}-4\gamma}}{2}$. Maximizing total activity satisfies Assumption \ref{ass:preference} since
\begin{equation*}
\mathbf{1}'\mathbf{x}^{\ast}=\sum_{t=0}^{\infty}\frac{\beta_{1}^{t+1}-\beta_{2}^{t+1}}{\beta_{1}-\beta_{2}}\mathbf{1}'\mathbf{G}^{t}\mathbf{1}.
\end{equation*}
\end{enumerate}

Our framework captures optimization problems across diverse economic contexts by allowing both the planner's farsightedness $D(t)$ and the network ranking measure $u(t)$ to vary with the application. In development economics, the framework can determine optimal transportation infrastructure when a city's GDP depends on road connectivity. Links represent physical roads connecting communities, and the planner maximizes economic output by strategically choosing which roads to build in each period. In information transmission, the framework applies when a social planner constructs a communication network to maximize information diffusion, as measured by aggregate diffusion centrality \citep{Banerjee2013,Bramoulle2024}. Links here represent communication channels through which information flows. In industrial organization, when agents interact strategically, the framework characterizes optimal peer effect structures for maximizing aggregate equilibrium activity or welfare, where links represent strategic interactions between agents.

However, Assumption \ref{ass:preference} may not hold in all economically relevant settings. Consider the planner's instantaneous utility
\begin{equation*}
u(\mathbf{G})=\mathbf{1}'(\mathbf{I}+\phi\mathbf{G})^{-1}\mathbf{1}=\sum_{k=0}^{\infty}(-\phi)^{k}W^{k}(\mathbf{G})
\end{equation*}
where $\phi>0$. The alternating signs of the coefficients mean that even when $W^{k}(\mathbf{G})>W^{k}(\mathbf{\hat{G}})$ for all $k\geq 2$, we cannot conclude that $u(\mathbf{G})>u(\mathbf{\hat{G}})$, violating Assumption \ref{ass:preference}. Such non-monotonic preferences arise naturally in settings with strategic substitution, where increased connectivity can reduce equilibrium welfare \citep{Bramoulle2007,Bramoulle2014,Elliott2019}. Our framework also does not apply to production networks \citep{Acemoglu2012,acemoglu2020,elliott2022}, which are inherently directed and whose general equilibrium outcomes are not monotonic in walks of all lengths. Extending our approach to directed networks with objectives non-monotonic in walks remains an important direction for future research.

\subsection{Notations}

We conclude the model setup by introducing special network structures that play an essential role in our analysis. Denote $N_{i}(\mathbf{G})=\{j:g_{ij}=1\}$ as the set of $i$'s neighbors in network $\mathbf{G}$.

\begin{defin}
\label{def-NSG} A network $\mathbf{G}$ is called a nested split graph (NSG)\footnote{In the graph theory literature, several equivalent definitions of NSGs exist; see \cite{N.V.R.Mahadev1995} for alternative characterizations and properties. See \cite{Koenig2014}, \cite{Billand2015}, \cite{Belhaj2016}, and \cite{BILLAND2023} for economic applications.} if, for each $i\neq j$, either 
\begin{equation*}
N_{i}(\mathbf{G})\backslash\{j\}\subseteq N_{j}(\mathbf{G})\backslash\{i\}\quad\text{or}\quad N_{j}(\mathbf{G})\backslash\{i\}\subseteq N_{i}(\mathbf{G})\backslash\{j\}.
\end{equation*}
\end{defin}

\begin{figure}[ht]
\centering
\includegraphics[scale=0.5]{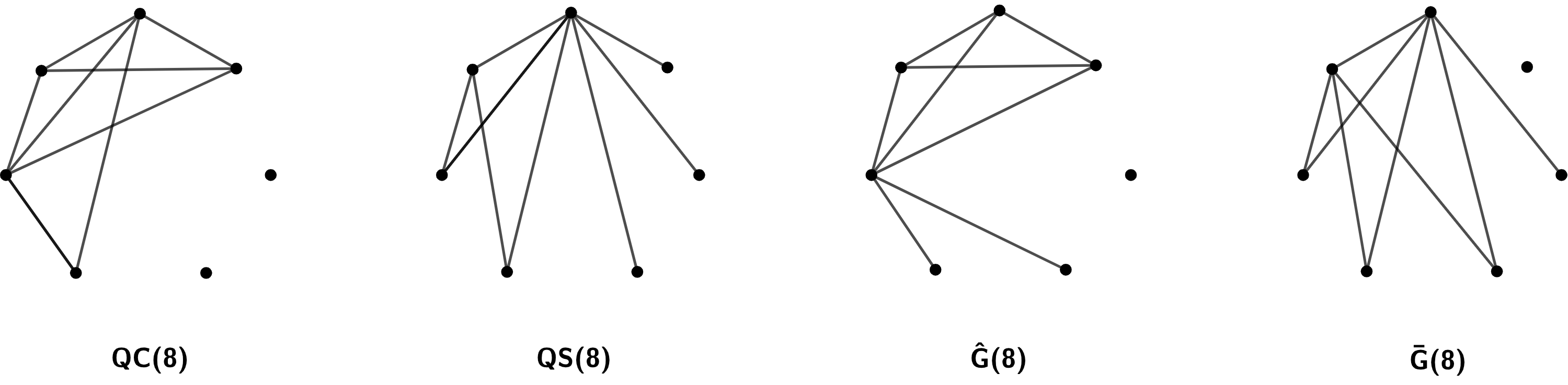}
\caption{Nested split graphs in $\mathcal{NSG}(8)$ with $n=7$ nodes and $k=8$ links}
\label{fig:NSGs}
\end{figure}

For any positive integer $k$, let $\mathcal{NSG}(k)$ denote the set of all NSGs with $k$ links. Figure \ref{fig:NSGs} illustrates all possible NSGs with $k=8$ links on $n=7$ nodes. NSGs represent a rich family of network structures with diverse topological properties.

\begin{defin}
\label{def-QC} A network $\mathbf{G}\in\mathcal{NSG}(t)$ is called a quasi-complete graph, denoted by $\mathbf{QC}(t)$, if it contains a clique of size $p$, where 
\begin{equation*}
\frac{p(p-1)}{2}\leq t<\frac{p(p+1)}{2},
\end{equation*}
and the remaining $t-\frac{p(p-1)}{2}$ links connect one additional node to nodes in the clique.
\end{defin}

\begin{figure}[ht]
\centering
\includegraphics[scale=0.6]{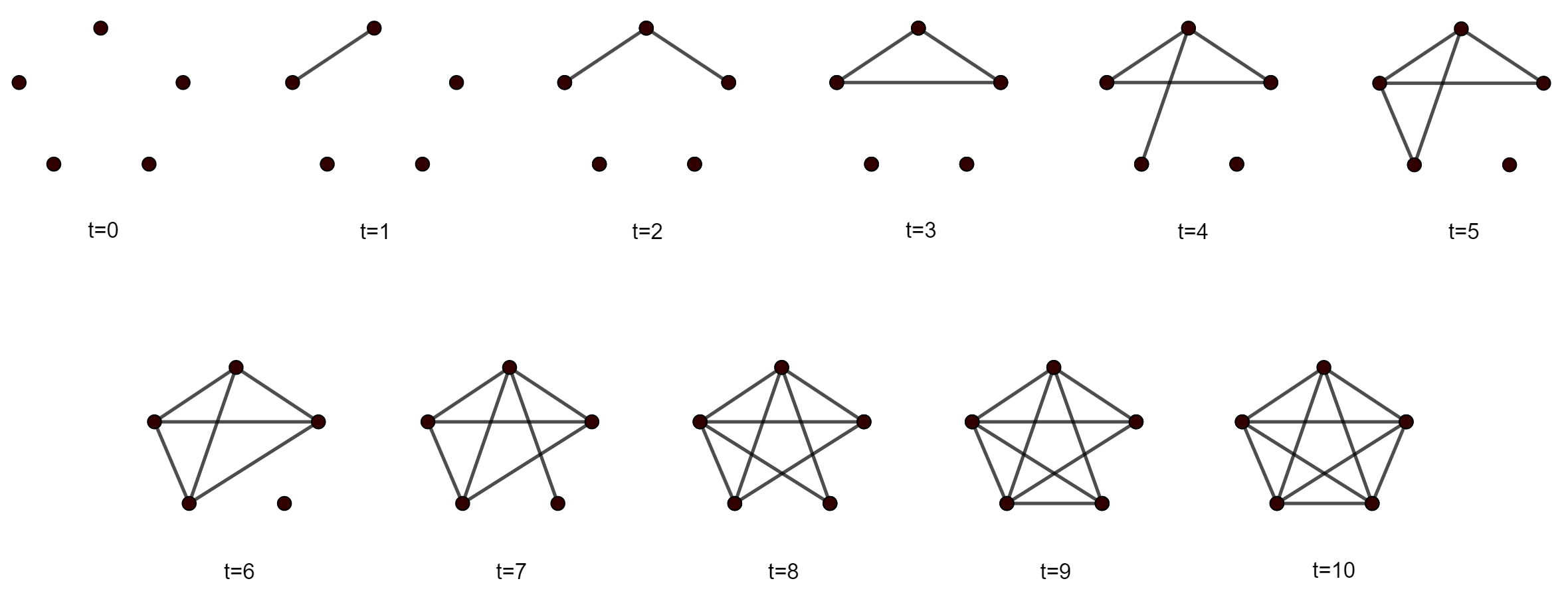}
\caption{Quasi-complete networks with $n=5$ nodes.}
\label{fig:QC}
\end{figure}

The quasi-complete (QC) graph represents a specialized subclass of NSGs characterized by having the largest possible clique among all graphs with a given number of links. Specifically, for a QC graph with $t$ links satisfying $\frac{p(p-1)}{2}\leq t<\frac{p(p+1)}{2}$, there exists a unique maximal clique of size $p$. Figure \ref{fig:QC} displays QC graphs with $n=5$ nodes for varying numbers of links $t$.Notably, for a fixed number of links, the QC graph is unique up to isomorphism. Figure \ref{fig:QC} also illustrates a dynamic network formation process across $T=10$ periods.

\section{Main Results}
\label{sec:main results}

\subsection{The Optimal Formation Process}
\label{sub:results}

\begin{thm}
\label{thm:unweighted-NSG}~

\begin{enumerate}
\item An optimal path $\mathbf{s}^*$ exists. Moreover, there always exists an optimal path $\mathbf{s}^{\ast}=(\mathbf{G}^{\ast}(t))_{t=1}^{T}$ such that $\mathbf{G}^{\ast}(t)\in\mathcal{NSG}(t)$ for all $1\leq t\leq T$;

\item If $D(r)>0$ for some period $r$, then any optimal path $\mathbf{s}^{\ast}=(\mathbf{G}^{\ast}(t))_{t=1}^{T}$ that solves \eqref{design problem} must satisfy $\mathbf{G}^{\ast}(r)\in\mathcal{NSG}(r)$.
\end{enumerate}
\end{thm}

Theorem \ref{thm:unweighted-NSG} establishes that NSGs characterize optimal network formation at each stage of the process. Part (i) shows that without loss of optimality, we can restrict attention to paths that form an NSG in every period. Part (ii)  states that if the planner assigns positive weight to some period $r$, then the network formed in that period must be an NSG along any optimal path. When applied to the instantaneous utility $u(\mathbf{G})=\lambda_{\max}(\mathbf{G})$, our result complements the literature on maximal spectral radius by incorporating a dynamic formation process. Specifically, any formation process that maximizes a weighted sum of spectral radii across periods necessarily produces an NSG in each period.

Let $\mathcal{G}(K)$ denote the set of all networks with $K$ links, and $\lambda_{\max}(K)$ denote the maximum spectral radius over all networks in $\mathcal{G}(K)$. Applying Theorem \ref{thm:unweighted-NSG} to the far-sighted planner (see Definition \ref{def-FS}) with $u=b(\alpha,\phi,\mathbf{G})$, we obtain:

\begin{proposition}
\label{prop-alpha}When $\phi\in[0,\frac{1}{\lambda_{\max}(K)})$, any solution to 
\begin{equation}
\max_{\mathbf{G}\in\mathcal{G}(K)}b(\alpha,\phi,\mathbf{G})
\label{program-KB-alpha}
\end{equation}
is an NSG for any positive integer $\alpha$.
\end{proposition}

The special cases $\alpha=1$ and $\alpha=2$ of Problem \eqref{program-KB-alpha} are studied in \cite{Belhaj2016}.

\begin{corollary}
[\citealt{Belhaj2016}] \label{cor-TE2016} The solution to 
\begin{equation*}
\text{either }\max_{\mathbf{G}\in\mathcal{G}(K)}b(1,\phi,\mathbf{G})\text{ or }\max_{\mathbf{G}\in\mathcal{G}(K)}b(2,\phi,\mathbf{G})
\end{equation*}
is an NSG.
\end{corollary}

\cite{Belhaj2016} show that when the planner maximizes either the sum of Katz-Bonacich centrality or the sum of its square, the optimal network in $\mathcal{G}(K)$ must be an NSG. Theorem \ref{thm:unweighted-NSG} and Proposition \ref{prop-alpha} extend their findings in two key dimensions. First, our results apply to any discount function, not merely those assigning positive weight only to the final period. This ensures that every intermediate network along the optimal formation path must be an NSG, not just the terminal network. Second, we accommodate a broader class of planner preferences beyond Katz-Bonacich centrality, encompassing diffusion centrality, community centrality, and spectral radius as discussed in Section \ref{sub:examples}.

We sketch the proof of Theorem \ref{thm:unweighted-NSG}, which proceeds in two main steps. In the first step, we examine how the neighbor reallocation operation introduced by \cite{Belhaj2016} affects aggregate walks of various lengths.

\begin{defin}
For any two nodes $i,j\in N$, define the operator 
\begin{equation*}
\mathcal{T}_{j\rightarrow i}:\mathcal{G}(K)\rightarrow\mathcal{G}(K)
\end{equation*}
such that $\mathcal{T}_{j\rightarrow i}(\mathbf{G})=\mathbf{G}+\sum_{l\in L}\mathbf{E}_{il}-\sum_{l\in L}\mathbf{E}_{jl}$ for any $\mathbf{G}\in\mathcal{G}(K)$, where 
\begin{equation*}
L:=\{l\in N\backslash\{i,j\}:g_{il}=0\text{ and }g_{jl}=1\}
\end{equation*}
is the set of nodes that are neighbors of $j$ but not of $i$.
\end{defin}

The network $\mathcal{T}_{j\rightarrow i}(\mathbf{G})$ is obtained by shifting \emph{all} neighbors of $j$ that are not connected to $i$ to become neighbors of $i$ instead. Three properties of the operator $\mathcal{T}_{j\rightarrow i}$ follow directly from the definition:

\begin{enumerate}
\item $\mathcal{T}_{j\rightarrow i}$ only reallocates links and preserves the total number of links.

\item $\mathcal{T}_{i\rightarrow j}(\mathbf{G}) \cong \mathcal{T}_{j\rightarrow i}(\mathbf{G})$.

\item If $\mathbf{G}\in\mathbb{S}(\mathbf{\tilde{G}})$, then $\mathcal{T}_{j\rightarrow i}(\mathbf{G})\in\mathbb{S}(\mathcal{T}_{j\rightarrow i}(\mathbf{\tilde{G}}))$.
\end{enumerate}

\begin{lem}
\label{lem:wel-increase}When $L\neq\emptyset$ and $\mathcal{T}_{j\rightarrow i}(\mathbf{G})\ncong \mathbf{G}$, then $W^{k}(\mathbf{G})<W^{k}(\mathcal{T}_{j\rightarrow i}(\mathbf{G}))$ for any integer $k\geq 2$ and $\lambda_{\max}(\mathbf{G})<\lambda_{\max}(\mathcal{T}_{j\rightarrow i}(\mathbf{G}))$.
\end{lem}

\cite{Belhaj2016} showed that the operator $\mathcal{T}_{j\rightarrow i}$ improves both the sum of Katz-Bonacich centrality and the sum of squared Katz-Bonacich centrality. Lemma \ref{lem:wel-increase} extends their result by demonstrating that this operation enhances any walk-based centrality measure satisfying Assumption \ref{ass:preference}. Furthermore, the corresponding lemma in \cite{Belhaj2016} required reallocating neighbors specifically from the node with lower Katz-Bonacich centrality to the node with higher centrality. Our Lemma \ref{lem:wel-increase} eliminates this directional constraint by the second property of $\mathcal{T}_{j\rightarrow i}$.\footnote{The reverse direction would involve reallocating neighbors from the node with higher Katz-Bonacich centrality to the node with lower centrality.}

A straightforward corollary of Lemma \ref{lem:wel-increase} is that in static network design, if the planner's preference satisfies Assumption \ref{ass:preference}, then the optimal network must be an NSG. This follows because if a network is not an NSG, there always exists a utility-improving transformation $\mathcal{T}_{j\rightarrow i}$. However, Lemma \ref{lem:wel-increase} cannot be directly applied to sequential network design, since applying the neighbor reallocation operation to a single network along a feasible path does not necessarily yield another feasible path.

To address this challenge, in the second step of the proof, we develop the following algorithm to construct a perturbed path that maintains feasibility while guaranteeing utility improvement.

\begin{algorithm}
\label{alg:dynamic-reallocation} For any strategy $\mathbf{s}=(\mathbf{G}(t))_{t=1}^{T}$, we define the following algorithm:

\begin{description}
\item[\textbf{Step 1.}] Check whether $\mathbf{G}(t)\in\mathcal{NSG}(t)$ for all $1\leq t\leq T$.

\begin{itemize}
\item If true, the algorithm terminates;

\item If false, proceed to Step 2.
\end{itemize}

\item[\textbf{Step 2.}] Find ${t}^{\prime}$ such that $\mathbf{G}(t)\in\mathcal{NSG}(t)$ for all $t\leq{t}^{\prime}-1$ and $\mathbf{G}({t}^{\prime})\notin\mathcal{NSG}({t}^{\prime})$. Find a pair of nodes $i,j\in N$ that violates nestedness at ${t}^{\prime}$ (where $i$'s neighborhood contains $j$'s neighborhood in all periods $t < t'$).

\item[\textbf{Step 3.}] Construct another strategy $\hat{\mathbf{s}}=(\hat{\mathbf{G}}(t))_{t=1}^{T}$ according to the following rules:

\begin{itemize}
\item If $t<{t}^\prime$, let $\hat{\mathbf{G}}(t)=\mathbf{G}(t)$.

\item If $t\geq{t}^\prime$ and $\mathbf{G}(t+1)=\mathbf{G}(t)+\mathbf{E}_{jl}$ for some $l\notin\{i,j\}$, then:

\begin{enumerate}
\item[i.] If $\hat{g}_{il}(t)=0$, let $\hat{\mathbf{G}}(t+1)=\hat{\mathbf{G}}(t)+\mathbf{E}_{il}$ (e.g., $t=4,5$ in Figure \ref{fig:proof});

\item[ii.] If $\hat{g}_{il}(t)=1$, let $\hat{\mathbf{G}}(t+1)=\hat{\mathbf{G}}(t)+\mathbf{E}_{jl}$.
\end{enumerate}

\item If $t\geq{t}^\prime$ and $\mathbf{G}(t+1)=\mathbf{G}(t)+\mathbf{E}_{il}$ for some $l\notin\{i,j\}$, then:

\begin{enumerate}
\item[i.] If $\hat{g}_{il}(t)=0$, let $\hat{\mathbf{G}}(t+1)=\hat{\mathbf{G}}(t)+\mathbf{E}_{il}$;

\item[ii.] If $\hat{g}_{il}(t)=1$, let $\hat{\mathbf{G}}(t+1)=\hat{\mathbf{G}}(t)+\mathbf{E}_{jl}$.
\end{enumerate}

\item If $t\geq{t}^\prime$ and $\mathbf{G}(t+1)=\mathbf{G}(t)+\mathbf{E}_{lk}$ for some $l,k\notin\{i,j\}$ or $(l,k)=(i,j)$, then let $\hat{\mathbf{G}}(t+1)=\hat{\mathbf{G}}(t)+\mathbf{E}_{lk}$ (e.g., $t=6$ in Figure \ref{fig:proof}).
\end{itemize}

\item[\textbf{Step 4.}] Set $\mathbf{s}=\hat{\mathbf{s}}$ and return to \textbf{Step 1}.
\end{description}
\end{algorithm}

We illustrate the algorithm using a six-node example in Figure \ref{fig:proof}. In this example, the original sequence $\mathbf{s}$ produces the first non-NSG at period 4, where $(i,j)$ is a pair violating nestedness. In period 4, the original process $\mathbf{s}$ connects one node to node $j$. Since this node is not connected to $i$, the algorithm switches the link to connect $i$ and this node instead (the second bullet of Step 3). A similar operation is performed in period 5. In period 6, since the original process $\mathbf{s}$ connects a pair not involving either $i$ or $j$, no perturbation occurs (the fourth bullet of Step 3). In the final period, the original process $\mathbf{s}$ connects a node to $i$. Since this node is already connected to $i$ under the perturbed process, the algorithm switches this link to connect the node and $j$ (the third bullet of Step 3). 

This example illustrates a general property of the algorithm. For any sequence $\mathbf{s}$ where $(i,j)$ is the first non-nested pair at some period, the algorithm produces a new sequence that reallocates all of $j$'s neighbors that are not $i$'s neighbors---the set $L(t)$ shown by the dotted circle in Figure \ref{fig:proof}---to node $i$ in each period $t$. Therefore, the newly constructed sequence $(\hat{\mathbf{G}}(t))_{t=1}^{T}$ induces a weakly higher payoff $v(\cdot)$ than the original sequence $(\mathbf{G}(t))_{t=1}^{T}$ under Assumption \ref{ass:preference}. With each iteration of the algorithm, either the sequence remains unchanged, or the payoff $v(\cdot)$ weakly increases by transforming at least one non-NSG into an NSG. That is, when the input and output differ, there exists at least one period $t^{\prime}$ where $\mathbf{G}(t^{\prime})$ changes. Therefore, the algorithm must terminate. At termination, the resulting sequence consists entirely of NSGs.

\begin{figure}[ht]
\centering
\includegraphics[scale=0.47]{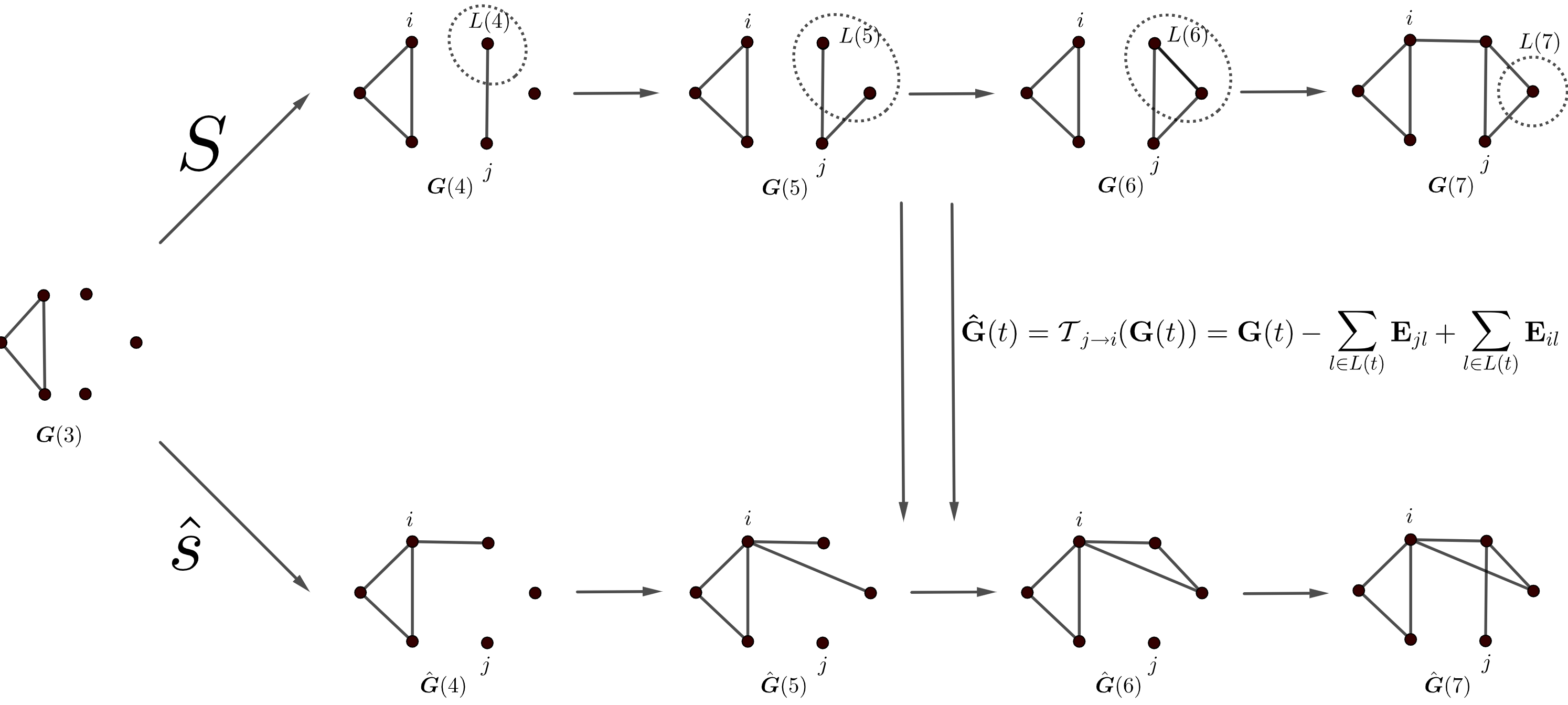}
\caption{An illustrative example for the proof of Theorem \protect\ref{thm:unweighted-NSG}}
\label{fig:proof}
\end{figure}

Before refining our characterization, we first clarify the scope of Theorem~1 and its connections to observed network structures. Our framework is most applicable in settings where distance costs are negligible. For instance, within an urban region, relative cost differences arising from geographic distance may be small compared with other constraints. Over longer distances between cities, however, geographic constraints become binding, and incorporating spatial factors remains an important direction for future research.\footnote{We thank an anonymous referee for pointing out this limitation of our framework.} 

Theorem~1 is consistent with empirical findings showing that nestedness arises across various network environments in which geographical constraints are less dominant. \cite{Uzzi1996} find that the organizational structure of the New York garment industry exhibits nestedness, reflecting patterns of collaboration and resource sharing within a geographically concentrated area. \cite{Bastos2018} demonstrate that the diffusion of specialized information on Twitter leads to a core-periphery architecture—a type of nested split graph—in a setting where physical distance is irrelevant.\footnote{Trade networks also frequently display nestedness in the context of directed networks. \citet{Akerman2014} show that the set of countries to which a large country exports arms almost always contains the set of countries to which a smaller country exports. Similarly, \citet{Ren2020} highlight the significance of nested trade network structures for high-complexity products.}

Theorem \ref{thm:unweighted-NSG} establishes the optimality of NSGs under general preference specifications. However, this generality also implies that the NSG-based characterization remains relatively coarse (recall that there may be multiple NSGs with the same total link count $T$). The next subsection provides a refined characterization of the optimal path when the planner is sufficiently myopic.

\subsection{Myopic Optimum}
\label{sub:myopic}

In this subsection, we characterize the optimal network formation path for a myopic planner. Unlike a far-sighted planner who optimizes over the entire formation path, a myopic planner follows a greedy algorithm:

\begin{defin}
\label{def:greedy strategy} 
A network formation path $\mathbf{\tilde{s}}=(\mathbf{\tilde{G}}(t))_{t=1}^{T}$ is induced by the greedy algorithm if, for any $t$,
\begin{equation*}
\mathbf{\tilde{G}}(t) \in \arg \max_{\mathbf{G}\in\mathbb{S}(\mathbf{\tilde{G}}(t-1))}u(\mathbf{G}).
\end{equation*}
\end{defin}

The greedy algorithm is intuitive: at each step, add the link that maximizes current utility. The following theorem fully characterizes the network path induced by this algorithm.

\begin{thm}
\label{thm:unweighted-QC} 
A myopic planner produces the same outcome as the greedy algorithm, $\mathbf{\tilde{s}}$. Furthermore, $\mathbf{\tilde{s}}=(\mathbf{\tilde{G}}(t))_{t=1}^{T}$ corresponds to a quasi-complete graph in each period, i.e., $\mathbf{\tilde{G}}(t)\cong \mathbf{QC}(t)$ for all $t\leq T$.
\end{thm}

Theorem \ref{thm:unweighted-QC} characterizes the network path induced by the greedy algorithm. The algorithm produces a unique sequence of networks up to isomorphism for any number of total links $T$, subject to the mild restriction on planner preferences in Assumption \ref{ass:preference}. 

This theorem has two important implications. First, from an economic perspective, when a planner heavily discounts future utility streams, the resulting network formation will follow a sequence of QC graphs. This myopic planning scenario is common in practice. For example, a mayor tasked with building roads to connect separated villages may primarily focus on economic outcomes during their term of office. Second, from a network theory perspective, this result provides a microeconomic foundation for QC graphs by demonstrating their emergence through the natural process of greedy optimization.

To prove Theorem \ref{thm:unweighted-QC}, we restrict our attention to formation paths that induce NSGs at each period by Theorem \ref{thm:unweighted-NSG}. Assume that a QC graph is formed at period $t-1$. After adding one more link to this QC graph, one can obtain a unique NSG (up to isomorphism) if $t-1=p(p-1)/2$ for some integer $p$. However, if $t-1\neq p(p-1)/2$, two non-isomorphic NSGs can result. Figure \ref{fig:QCproof} illustrates the latter case: one NSG network $\mathbf{QC}=\mathbf{G}+\mathbf{E}_{34}$ is quasi-complete, and the second NSG is $\mathbf{\hat{G}}=\mathbf{G}+\mathbf{E}_{15}$.

\begin{figure}[tph]
\begin{center}
\includegraphics[scale=0.9]{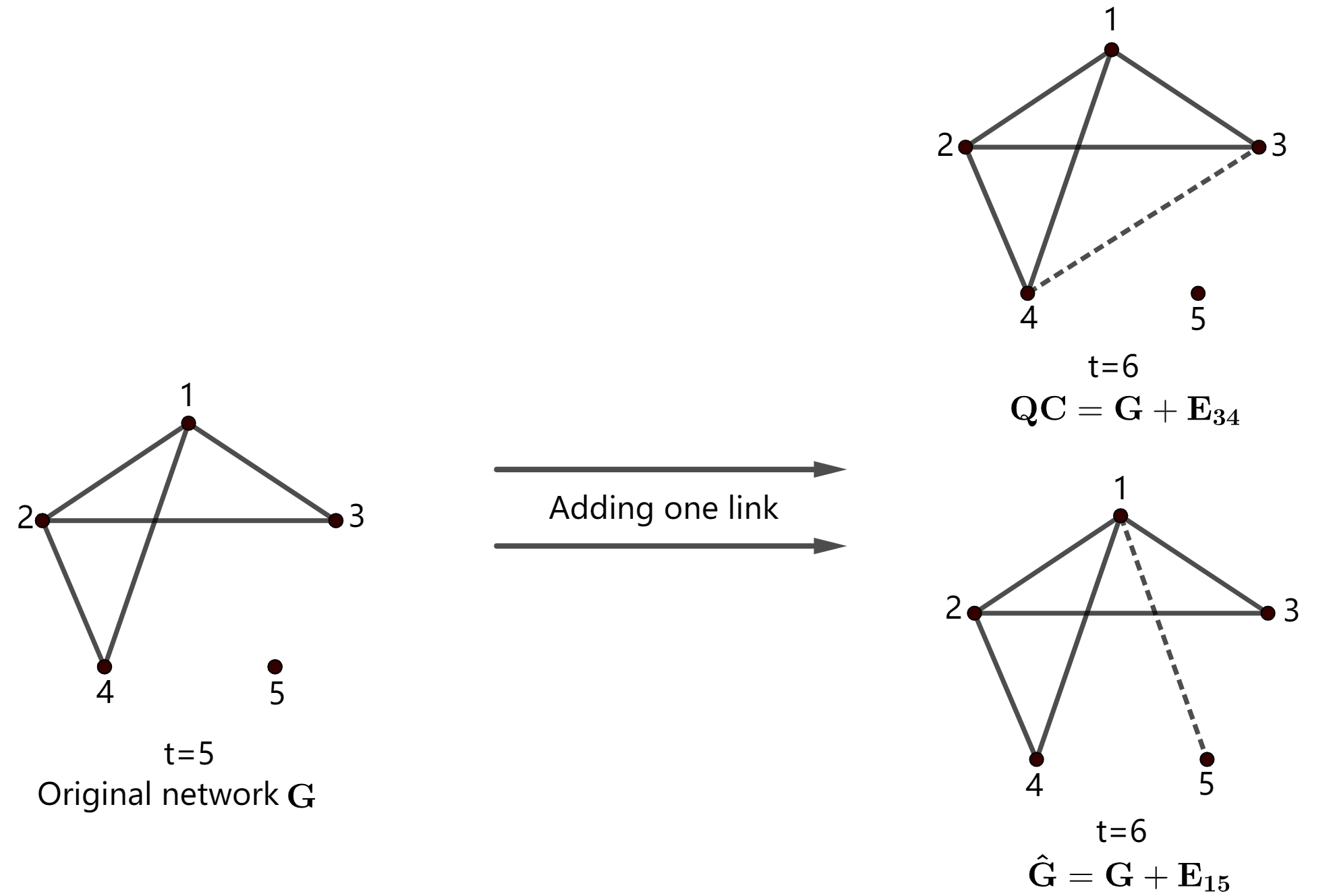}
\end{center}
\caption{Two different NSGs succeeding a quasi-complete graph}
\label{fig:QCproof}
\end{figure}

A key step in the proof is to show that the QC graph dominates the other NSG $\mathbf{\hat{G}}$ in terms of both aggregate walks of any length and the spectral radius.

\begin{lem}
\label{lem:unweighted-QC} 
Suppose $t-1\neq p(p-1)/2$ for some integer $p$. Then $W^{k}(\mathbf{QC})>W^{k}(\mathbf{\hat{G}})$ for any $k\geq 2$ and $\lambda_{\max}(\mathbf{QC})>\lambda_{\max}(\mathbf{\hat{G}})$.
\end{lem}

To prove the first part of Lemma \ref{lem:unweighted-QC}, we partition nodes into two categories based on whether their neighborhoods differ between the two NSGs. In Figure \ref{fig:QCproof}, nodes 1, 3, 4, and 5 belong to the category whose neighborhoods differ between the two NSGs, while node 2 belongs to the category whose neighborhood remains identical across both NSGs. We analyze the combined walk counts of the two categories separately and show that the total number of walks of length $k$ originating from nodes in each category in network $\mathbf{QC}$ exceeds that in $\mathbf{\hat{G}}$. 

The second part, $\lambda_{\max}(\mathbf{QC})>\lambda_{\max}(\mathbf{\hat{G}})$, requires a different technique. From the first part, $W^{k}(\mathbf{QC})>W^{k}(\mathbf{\hat{G}})$ for any $k\geq 2$ implies $\lambda_{\max}(\mathbf{QC})\geq\lambda_{\max}(\mathbf{\hat{G}})$. We then use the eigencentrality equation $\lambda_{\max}(\mathbf{G})\mathbf{z}=\mathbf{Gz}$ to show that $\lambda_{\max}(\mathbf{QC})\neq\lambda_{\max}(\mathbf{\hat{G}})$, where $\mathbf{z}$ is the vector of eigencentralities.

The following result is immediate from Lemma \ref{lem:unweighted-QC}.

\begin{corollary}
Suppose $t-1\neq p(p-1)/2$ for some integer $p$. There exist two NSGs, $\mathbf{G}$ and $\mathbf{G}^{\prime}$ in $\mathcal{G}(t)$, such that for any positive integer $\alpha$, 
\begin{equation*}
b(\alpha, \phi, \mathbf{G}) > b(\alpha, \phi, \mathbf{G}^{\prime}).
\end{equation*}
\end{corollary}

This result strengthens Corollary \ref{cor-TE2016}, and hence the main result of \cite{Belhaj2016}, by showing that within the set of NSGs, certain networks are never optimal for Problem \eqref{program-KB-alpha}. Previous literature on static network design typically stopped at showing that globally efficient networks must belong to the set of NSGs. However, such characterizations have limited ability to distinguish among NSGs. Lemma \ref{lem:unweighted-QC} takes a first step toward further refinement within the class of NSGs, though complete discrimination among all NSGs remains an open question for future research.

Finally, it is worth noting that the greedy strategy's focus on short-term gains may come at the cost of long-term efficiency. By ignoring the potential future benefits of currently sub-optimal networks, the myopic planner may miss opportunities to create more efficient network structures in the long run. The following example demonstrates the difference between networks formed by myopic and farsighted planners.

\begin{ex}
\label{ex:1}
Consider the planner's problem (\ref{design problem}) with 7 nodes and 8 periods, i.e., $n=7$ and $T=8$. By Theorem \ref{thm:unweighted-NSG}, the finally formed network must be one of the four NSGs listed in Figure \ref{fig:NSGs}. Suppose the planner is farsighted and cares about the aggregate square of KB centrality, i.e., $v(\mathbf{s})=b(2, \phi,\mathbf{G}(T))$. Table \ref{tab:Comparedelta} lists $b(2, \phi,\mathbf{G})$ induced by these four NSGs when $\phi=0.01$.

\begin{table}[!ht]
\centering
\resizebox{10cm}{!}{
\begin{tabular}{|l|l|l|l|l|}
\hline
  &$\mathbf{QC}(8)$ & $\mathbf{QS}(8)$&$\mathbf{\hat{G}}(8)$&$\mathbf{\bar{G}}(8)$ \\ \hline
$b(2, \phi,\mathbf{G})$ & 7.3370 & 7.3374* & 7.3368 & 7.3362 \\ \hline
\end{tabular}}
\caption{Comparison among different NSGs}
\label{tab:Comparedelta}
\end{table}

From this table, we observe that the optimal network for a farsighted planner is the quasi-star $\mathbf{QS}(8)$, while for a myopic planner, it is the quasi-complete graph $\mathbf{QC}(8)$, as per Theorem \ref{thm:unweighted-QC}.
\end{ex}

Example \ref{ex:1} highlights two important insights. First, it confirms the established wisdom that greedy algorithms do not necessarily yield globally optimal solutions in network design problems. Second, it demonstrates that the optimal network formation path depends critically on the planner's time preferences, with different discount factors leading to different network structures.

\subsection{Comparative Statics}

We assume the discount factor follows a geometric pattern with factor $\delta>0$. The planner's value function is then given by $v(\mathbf{s}):=\sum_{t=1}^{T}\delta^{t}u(\mathbf{G}(t))$. The planner's optimization problem is: 
\begin{equation}
\max_{\mathbf{s}\in S}v(\mathbf{s}).  \label{eq:CSproblem}
\end{equation}
In this section, we focus on two special forms of instantaneous utility: $u(\mathbf{G})=b(\alpha,\phi,\mathbf{G})$ and $u(\mathbf{G})=c(\beta,\mathbf{G})$ as discussed in Section \ref{sub:examples}. We analyze how optimal network structures vary with these parameters in extreme cases.

\begin{defin}
\label{QS}
A network $\mathbf{G}\in\mathcal{NSG}(t)$ is a quasi-star graph, denoted by $\mathbf{QS}(t)$, if it has a set of $p$ central nodes with $n-1$ links each, and the remaining $t-p(n-1)$ links are allocated to construct another central node.
\end{defin}

\begin{figure}[h]
\centering
\includegraphics[scale=0.6]{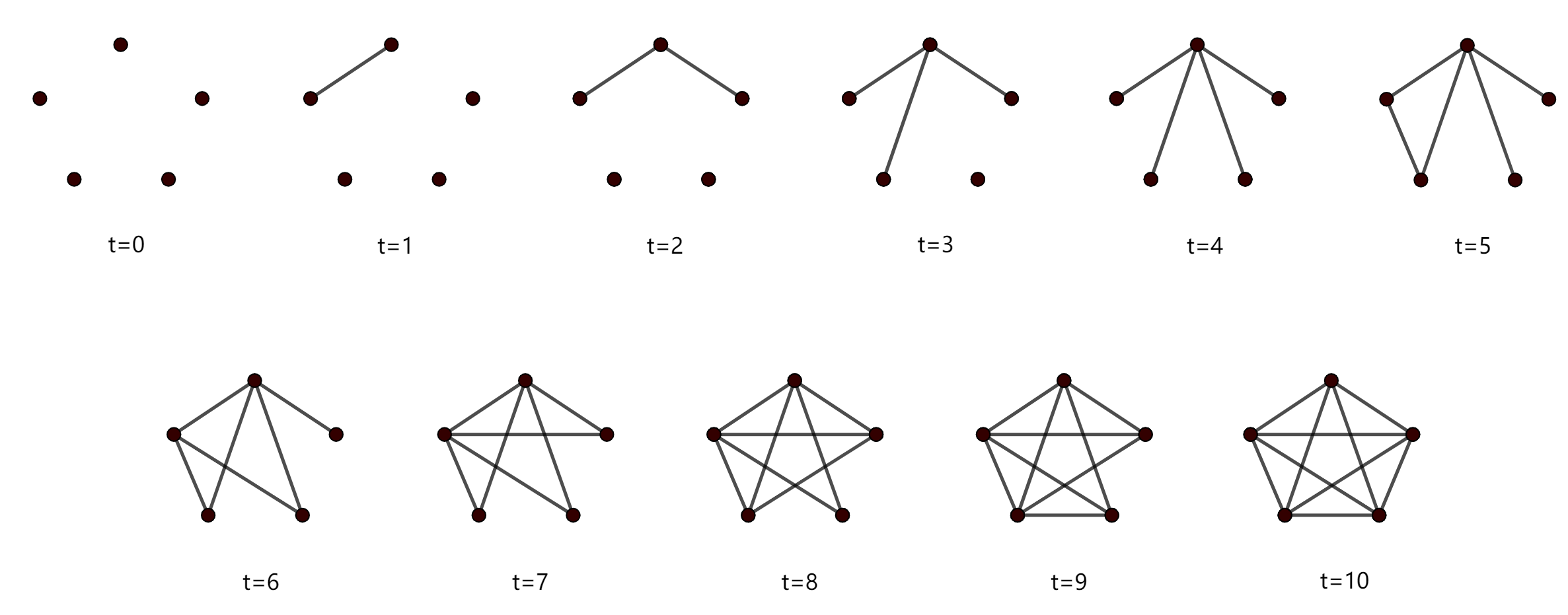}
\caption{Quasi-star networks.}
\label{fig:QS}
\end{figure}

The quasi-star is another prominent subclass of NSGs that maximizes the number of nodes with degree $n-1$. Figure \ref{fig:QS} illustrates quasi-star graphs with 5 nodes and various numbers of links. Notably, $\mathbf{QS}(t)$ is the graph complement of $\mathbf{QC}(t^{\prime})$ where $t^{\prime}=\frac{n(n-1)}{2}-t$. That is, the sum of the two adjacency matrices $\mathbf{QS}(t)+\mathbf{QC}(t^{\prime})$ equals the adjacency matrix of a complete graph after permutation.

\begin{corollary}
\label{cor:comparative} 
Suppose the number of nodes $n\geq 6$. Then the following holds for the solution to problem \eqref{eq:CSproblem}.

\begin{enumerate}
\item Suppose $u(\mathbf{G})=b(\alpha,\phi,\mathbf{G})$.

\begin{enumerate}
\item There exist $\bar{\delta}, \underline{\phi}>0$ such that, for all $\delta\geq\bar{\delta}$ and $\phi\leq\underline{\phi}$, the optimal solution at period $T$, $\mathbf{G}^{\ast}(T)$, is a quasi-star network when $3<T<\frac{n^{2}-3n}{4}$ and is a quasi-complete network when $\frac{n^{2}+n}{4}<T\leq\frac{n(n-1)}{2}$.

\item There exists $\underline{\delta}>0$ such that, for all $\delta\leq\underline{\delta}$, the optimal solution at $t$, $\mathbf{G}^{\ast}(t)$, is a quasi-complete network for any value of $\phi$ and $t$.
\end{enumerate}

\item Suppose $u(\mathbf{G})=c(\beta,\mathbf{G})$.

\begin{enumerate}
\item There exist $\bar{\delta}, \underline{\beta}>0$ such that, for all $\delta\geq\bar{\delta}$ and $\beta\leq\underline{\beta}$, the optimal solution at period $T$, $\mathbf{G}^{\ast}(T)$, is a quasi-star network when $3<T<\frac{n^{2}-3n}{4}$ and is a quasi-complete network when $\frac{n^{2}+n}{4}<T\leq\frac{n(n-1)}{2}$.

\item There exists $\bar{\beta}>0$ such that, for all $\beta\geq\bar{\beta}$, the solution to problem \eqref{eq:CSproblem} is also a solution to the dynamic maximal spectral radius problem $\max_{\mathbf{s}\in S}v(\mathbf{s})=\sum_{t=1}^{T}\delta^{t}\lambda_{\max}(\mathbf{G}(t))$.

\item There exists $\underline{\delta}>0$ such that, for all $\delta\leq\underline{\delta}$, the optimal solution at $t$, $\mathbf{G}^{\ast}(t)$, is a quasi-complete network for any value of $\beta$ and $t$.
\end{enumerate}
\end{enumerate}
\end{corollary}

This corollary demonstrates how the optimal network structure depends critically on the planner's time preference $\delta$, the specific form of instantaneous utility (through parameters $\phi$ and $\beta$), and the total number of formation periods $T$.

When the planner is far-sighted ($\delta\rightarrow+\infty$) and the parameters $\phi$ (when $u(\mathbf{G})=b(\alpha,\phi,\mathbf{G})$) or $\beta$ (when $u(\mathbf{G})=c(\beta,\mathbf{G})$) are sufficiently small, the optimal network structure transitions from quasi-star to quasi-complete as $T$ increases, as shown in the first parts of Corollary \ref{cor:comparative} for both utility forms. This transition occurs because, for a fixed total number of links $K$, as $\phi\rightarrow 0^{+}$, we have 
\begin{equation*}
\arg\max_{\mathbf{G}\in\mathcal{G}(K)}b(\alpha,\phi,\mathbf{G})=\arg\max_{\mathbf{G}\in\mathcal{G}(K)}c(\phi,\mathbf{G})=\arg\max_{\mathbf{G}\in\mathcal{G}(K)}\mathbf{1}^{\prime}\mathbf{G}^{2}\mathbf{1}
\end{equation*}
by equations \eqref{eq:KB&walks} and \eqref{eq:C&walks}. The optimization problem thus reduces to identifying the graph with the largest sum of squared degrees, and our results follow directly from \cite{Abrego2009}.

The second part of the corollary for $u(\mathbf{G})=c(\beta,\mathbf{G})$ addresses the extreme case where $\beta$ is sufficiently large. In this case, networks maximizing $c(\beta,\mathbf{G})$ coincide with spectral radius maximizers: 
\begin{equation*}
\arg\max_{\mathbf{G}\in\mathcal{G}(K)}\lim_{\beta\rightarrow+\infty}c(\beta,\mathbf{G})=\arg\max_{\mathbf{G}\in\mathcal{G}(K)}\lambda_{\max}(\mathbf{G}).
\end{equation*}
Consequently, the planner's problem reduces to dynamically maximizing the spectral radius.

Finally, the last parts of the corollary for both instantaneous utility forms follow directly from Theorem \ref{thm:unweighted-QC}, which establishes that the myopic optimum ($\delta\rightarrow 0^+$) is always a quasi-complete network at each period, provided the instantaneous utility satisfies Assumption \ref{ass:preference}.

\section{Weighted Networks}

\label{sec:weighted}

In this section, we extend our main results to settings that involve weighted networks. The distinction between weighted and unweighted networks rests on whether link quality or capacity constitutes the primary decision variable. Weighted networks are particularly relevant when planners can determine not only which nodes to connect but also the quality or strength of each connection. For example, in highway infrastructure, planners may decide both which highways to construct and how many lanes each should contain. In power grid design, decisions involve not only which transmission lines to install but also their voltage or transmission capacity.

For simplicity, we assume that the weight between any pair of nodes ranges from zero to one. Instead of adding a single discrete link in each period, the planner now allocates one unit of total link weight across multiple connections. Note that the planner’s choice set is therefore continuous (and hence infinite).

Previously, Lemma \ref{lem:wel-increase} played a crucial role in
demonstrating the optimality of NSGs. However, its counterpart for weighted
networks does not hold, as illustrated in Example \ref{ex:weighted} below.
Consequently, the optimal networks may not be (weighted) NSGs if we impose
only Assumption \ref{ass:preference} on the planner's instantaneous utility.

\begin{ex}
\label{ex:weighted} 
\begin{figure}[h]
\centering
\par
\includegraphics[scale=1]{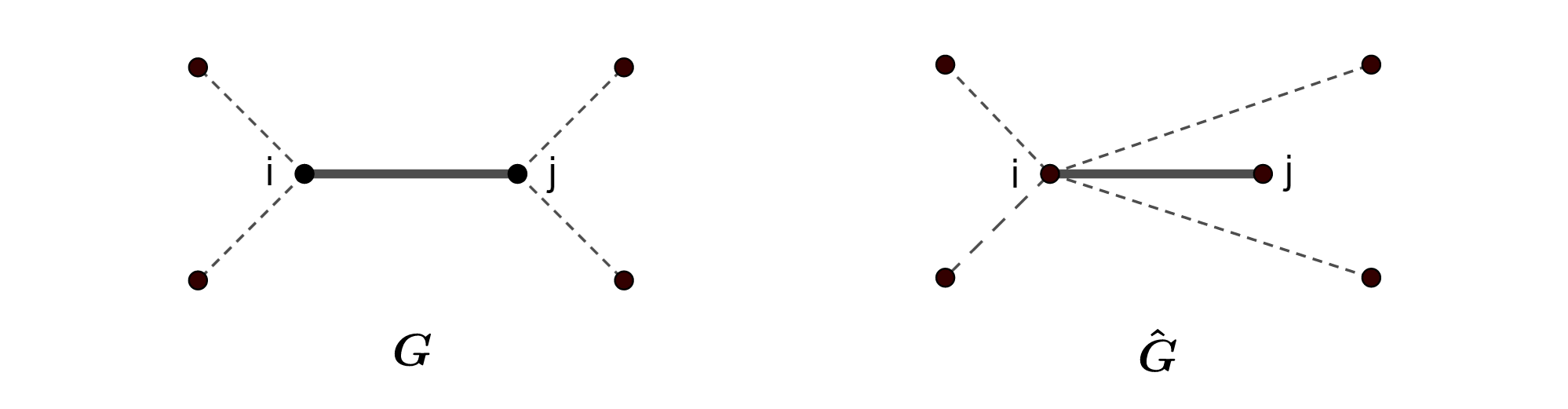}
\caption{A weight reallocation in weighted network}
\label{fig:appex}
\end{figure}
Consider the two undirected-weighted networks illustrated in Figure \ref%
{fig:appex}. The bold line and dashed line represent weights 1 and 0.1,
respectively. Note that network $\mathbf{\hat{G}}$ is obtained from $\mathbf{%
G}$ by shifting all of $j$'s neighbors to $i$, the same operation as in
Lemma \ref{lem:wel-increase}. Table \ref{tab:Compareww} compares the
aggregate (weighted) walks of each network with different lengths. 
\begin{table}[!ht]
\centering
\resizebox{6cm}{!}{
\begin{tabular}{|l|l|l|l|}
\hline
 &$k=2$ & $k=3$&$k=4$ \\ \hline
$W^k(\mathbf{G})$ & 2.92 & 2.976* & 3.0344 \\ \hline
$W^k(\hat{\mathbf{G}})$ & 3* & 2.912 & 3.12* \\ \hline
\end{tabular}}
\caption{Comparison of total walks}
\label{tab:Compareww}
\end{table}

Table \ref{tab:Compareww} illustrates that a weight reallocation from $j$ to 
$i$, such that $i$ weight-dominates $j$, may not increase the total number
of walks of a certain length: $W^k({\mathbf{G}}) > W^k(\hat{\mathbf{G}})$
when $k = 3$. This point highlights the failure of Lemma \ref%
{lem:wel-increase} in the context of weighted networks.

However, despite this inconsistency with walk counts, Table \ref%
{tab:CompareKB} illustrates that this neighborhood reallocation improves
both the aggregate KB centrality and its square: $b(1,\phi ,\mathbf{\hat{G}}%
)>b(1,\phi ,\mathbf{G})$ and $b(2,\phi ,\mathbf{\hat{G}})>b(2,\phi ,\mathbf{G%
})$ for the tested values of $\phi $.

\begin{table}[!ht]
\centering
\resizebox{10cm}{!}{
\begin{tabular}{|l|l|l|l|l|l|}
\hline
 &$\phi=0.1$ & $\phi=0.2$ &$\phi=0.3$ &$\phi=0.4$ &$\phi=0.5$ \\ \hline
 $b(1, \phi,\mathbf{G})$& 6.3125 & 6.7067 &7.2186&7.9088&8.8889 \\ \hline
  $b(1, \phi,\mathbf{\hat{G}})$& 6.3133* & 6.7095* &7.2246*&7.9194*&8.9054* \\ \hline
   $b(2, \phi,\mathbf{G})$& 6.6612 & 7.5977 &8.9825&11.1497&14.8148 \\ \hline
   
   $b(2, \phi,\mathbf{\hat{G}})$& 6.6634* & 7.6058* &9.0001*&11.1809*&14.8660* \\ \hline
\end{tabular}}
\caption{Comparison of total KB centrality}
\label{tab:CompareKB}
\end{table}
\end{ex}

In light of this example, we restrict instantaneous utility to either the
sum of KB centralities $b(1, \phi, \mathbf{G}(t))$ or the sum of the squares
of KB centralities $b(2, \phi, \mathbf{G}(t))$. This restriction is
justified on two grounds. First, as per equation (\ref{eq:KB&walks}), both $%
b\left(1, \phi, \mathbf{G}(t)\right)$ and $b\left(2, \phi, \mathbf{G}%
(t)\right)$ are weighted sums of aggregate walks of various lengths, which
aligns with the spirit of Assumption \ref{ass:preference}. Second, these
centrality measures have a strong connection to the network game literature.
When the network game is the classical linear-quadratic one introduced by
the seminal paper \cite{Ballester2006}, $b(1, \phi, \mathbf{G}(t))$ and $%
b(2, \phi, \mathbf{G}(t))$ represent the aggregate effort and utilitarian
welfare in equilibrium, respectively.

Throughout this section, the parameter $\phi$ is treated as a fixed constant
and is therefore omitted in the following for notational convenience.

Before proceeding with our analysis, we formally define the sequence of
successive weighted networks. Let $\mathcal{G}:=\{\mathbf{G}:g_{ij} =
g_{ji}\in [0,1], g_{ii} = 0, \forall i,j\in N\}$ represent the set of all
feasible weighted networks.

\begin{defin}
\label{def:weighted-feasibility} For any weighted network $\mathbf{G} \in 
\mathcal{G}$, denote 
\begin{equation*}
\mathbb{S}_{w}(\mathbf{G}):=\{\mathbf{\hat{G}} \in \mathcal{G}: \exists 
\mathbf{W}\geq \mathbf{0},\ \text{s.t.}\ \mathbf{W} = \mathbf{W}^\prime, 
\mathbf{1}^\prime \mathbf{W}\mathbf{1} = 2, \mathbf{\hat{G}} = \mathbf{G} + 
\mathbf{W}\},
\end{equation*}
the set of networks succeeding $\mathbf{G}$. Denote ${S}_{w}:=\{\mathbf{s}=(%
\mathbf{G}(t))^{T}_{t=1}|\mathbf{G}(t)\in \mathbb{S}_w(\mathbf{G}%
(t-1)),\forall t=1,\ldots, T\}$ the set of feasible network formation paths.%
\footnote{$\mathbf{G}(0)$ is the empty network.}
\end{defin}

We use the subscript "w" to distinguish the weighted network cases from
unweighted ones. Note that unweighted networks are just a special case of
weighted networks; therefore $S\subset S_{w}$. The flexibility in forming
weighted networks is expected to weakly improve the planner's utilities.

\subsection{Maximizing Aggregate Katz-Bonacich Centrality}

The problem of sequentially allocating a unit weight to maximize discounted
sum of KB centrality can be formulated as follows, 
\begin{equation}
\underset{\mathbf{s}\in S_{w}}{\max }\underset{t=1}{\overset{T}{\sum }}%
D\left(t\right)\cdot b\left( 1,\phi ,\mathbf{G}\left( t\right) \right) \text{%
.}  \label{design problem weighted 1}
\end{equation}
The main results, Theorems \ref{thm:unweighted-NSG} and \ref%
{thm:unweighted-QC}, can then be extended to weighted network design.

\begin{pro}
\label{pro:weighted-KB2-KB} If $\mathbf{s}_{w}^{\ast }=\left( \mathbf{G}%
^{\ast }\left(t\right) \right) _{t=1}^{T}$ is a solution to Problem %
\eqref{design problem weighted 1}, then $\mathbf{G}^{\ast }\left( t\right)$
is an (\textbf{unweighted)} NSG whenever $D\left( t\right) >0$. Furthermore,
when the planner is myopic, $\mathbf{G}^{\ast }\left( t\right) $ is (\textbf{%
unweighted)} quasi-complete.
\end{pro}

Proposition \ref{pro:weighted-KB2-KB} implies that the optimal path of
network formation results in an unweighted network at each period.
Therefore, the flexibility of forming weighted networks does not provide any
additional improvement to the planner, given their objective to maximize the
discounted sum of aggregate KB centralities.

The proof of Proposition \ref{pro:weighted-KB2-KB} relies on the following
lemma from \cite{Sun2023}, which demonstrates the convexity of aggregate KB
centrality with respect to the network structure.

\begin{lem}[Lemma A.2 in \citealt{Sun2023}]
\label{lem:convexity} Let $\mathcal{O}$ denote the set of $n\times n$
symmetric positive-definite matrices. Then, the function $V\left( \mathbf{A}%
\right) =\mathbf{1}^{\prime }\mathbf{A}^{-1}\mathbf{1}$ is convex in $%
\mathbf{A}\in \mathcal{O}$.
\end{lem}

Lemma \ref{lem:convexity} implies that the planner's utility, a weighted sum
of instantaneous utilities, is convex in the paths of network formation. The
remainder of the proof of Proposition \ref{pro:weighted-KB2-KB} is to show
that the set $S_w$ of feasible formation paths of weighted networks is a
convex set, and the set $S$ of feasible formation paths of unweighted
networks constitutes the extreme points of $S_w$.

\subsection{Maximizing  Aggregate Square of Katz-Bonacich Centrality}

\label{subsec:square-kb}

The problem of sequentially allocating unit weights to maximize the
discounted sum of squared KB centrality can be formulated as: 
\begin{equation}
\underset{\mathbf{s}\in S_{w}}{\max } \underset{t=1}{\overset{T}{\sum }}%
D\left(t\right)\cdot b\left(2, \mathbf{G}\left( t\right) \right)
\label{design problem weighted 2}
\end{equation}

Before presenting the main result of this subsection, we extend the
definition of NSG to weighted networks.

\begin{defin}
\label{def-weightedNSG} A weighted undirected network $\mathbf{G}$ is a
weighted nested split graph if for any two distinct nodes $i,j$, either $%
g_{ik}\geq g_{jk}$ $\forall k\notin \left\{ i,j\right\} $ or the converse.
\end{defin}

It can be easily verified that an (unweighted) NSG (as defined in Definition %
\ref{def-NSG}) satisfies the definition above.\footnote{%
In \cite{LI2023}, the concept of a generalized NSG is proposed for weighted
and directed networks, following the same spirit as Definition \ref%
{def-weightedNSG}.} The main result of this subsection is as follows:

\begin{pro}
\label{pro:weighted-KB2} The following holds,

\begin{enumerate}
\item[(i)] For any solution $\mathbf{s}_{w}^{\ast }=\left( \mathbf{G}^{\ast
}\left( t\right)\right) _{t=1}^{T}$ of Problem \ref{design problem weighted
2}, $\mathbf{G}^{\ast }\left( t\right) $ is a weighted NSG whenever $%
D\left(t\right)>0$. Moreover, for any node $i$, there is no two distinct
agents $j,k$ such that both $g_{ij}^{\ast}\left( t\right) $ and $%
g_{ik}^{\ast }\left( t\right) $ belong to $\left(0,1\right)$.

\item[(ii)] When the planner is myopic, $\mathbf{G}^{\ast }\left( t\right) $
is (\textbf{unweighted)} quasi-complete.
\end{enumerate}
\end{pro}

The first part of the proposition argues that, if the planner cares about
the instantaneous utility generated at period $t$, the network formed at
that period must be a weighted NSG. Furthermore, it demonstrates that no
node can maintain two distinct links with weights strictly between 0 and 1.
The second part shows that, if the planner is myopic, the network formed at
each period under the optimal path should be an (unweighted) QC graph.

\bigskip \noindent\textbf{The Proof Sketch of Proposition \ref%
{pro:weighted-KB2} (i)}. In proving the first part of Proposition \ref%
{pro:weighted-KB2}, we note that the sum of squared KB centralities is not
necessarily convex in the underlying network, making the techniques from the
previous subsection inapplicable. Instead, we have to invoke the following
lemma from \cite{Sun2023}, which can be viewed as an extension of Lemma \ref%
{lem:wel-increase} to weighted networks.

\begin{lem}[Proposition 4 in \citealt{Sun2023}]
\label{lem:sun2023} Consider two nodes $i,j$ in a weighted network $\mathbf{G%
}$ such that $b_{i}\left(1, \mathbf{G}\right) >b_{j}\left(1, \mathbf{G}%
\right) $. Suppose a weight reallocation from $j$ to $i$ is such that in the
post-reallocation network $\mathbf{\hat{G}}$, $\hat{g}_{ik}\geq \hat{g}_{jk}$
for any $k\notin \left\{ i,j\right\} $. Then, $b(2, \mathbf{\hat{G}})
>b\left(2, \mathbf{G}\right) $. Moreover, if $b_{i}\left(1, \mathbf{G}%
\right)<b_{j}\left(1, \mathbf{G}\right) $, a weight reallocation from $j$ to 
$i$ such that $\hat{g}_{ik}\geq g_{jk}$ for any $k\notin \left\{ i,j\right\} 
$ leads to $b(2, \mathbf{\hat{G}}) >b\left(2, \mathbf{G}\right) $.
\end{lem}

Analogous to Lemma \ref{lem:wel-increase}, Lemma \ref{lem:sun2023}
identifies a class of \emph{improving} weight-reallocation operations. While
the link-allocation operations in Lemma \ref{lem:wel-increase} represent a
subset of these weight-reallocation operations, the directionality of these
operations imposes distinct requirements depending on the nodes' relative
centrality. When reallocating weights from a node $j$ with lower KB
centrality to node $i$, the requirement is that node $i$ must uniformly
dominate node $j$ in the post-perturbation network. Conversely, when
reallocating from a node $j$ with higher KB centrality to node $i$, a more
stringent condition applies: node $i$'s weights in the perturbed network
must uniformly dominate node $j$'s weights in the original network. If we
restrict to unweighted networks, these two requirements turn out to be the
same and boil down to the requirement for link-reallocation operations.

A direct corollary of Lemma \ref{lem:sun2023} is that the optimal static
weighted network must be a weighted NSG. To apply this insight to sequential
weighted network design, we extend Algorithm \ref{alg:dynamic-reallocation}
to weighted networks (see the Appendix for details). In essence, if a path $%
\mathbf{s}_{w}=\left(\mathbf{G}\left(t\right)\right)_{t=1}^{T}$ does not
produce weighted NSGs at some period, with $t^\prime$ being the first such
period where neither node $i$ weight-dominates node $j$ nor vice versa, we
construct a new path. For each period $t\geq t^\prime$, we create $\mathbf{%
\hat{G}}\left(t\right)$ by transferring as much weight as possible from $j$
to $i$ in $\mathbf{G}\left(t\right)$, ensuring that either $\hat{g}%
_{ik}\left(t\right)=1$ or $\hat{g}_{jk}\left(t\right)=0$ for all $k\notin
\left\{i,j\right\}$. By Lemma \ref{lem:sun2023}, this new network $\mathbf{%
\hat{G}}\left(t\right)$ yields higher payoff than $\mathbf{G}\left(t\right)$%
, and the path $\hat{\mathbf{s}}=(\hat{\mathbf{G}}(t))^T_{t=1}$ remains
feasible.

Furthermore, Lemma \ref{lem:sun2023} excludes a significant subset of
weighted NSGs from optimality. Specifically, if a node $i$ has two links
with strictly intermediate weights ($0<g_{ij}\left(t\right),g_{ik}\left(t%
\right)<1$), transferring weight from link $ik$ to link $ij$ increases
payoff, even though $\mathbf{G}\left(t\right)$ is a weighted NSG.
Consequently, in an optimal network, no node can maintain two weighted links
with intermediate values.

The first part of Proposition \ref{pro:weighted-KB2} aligns with the main
finding in \cite{LI2023}, which showed that optimal complementary networks
should be weighted NSGs when the planner's utility function is
differentiable in nodes' equilibrium efforts. Our contribution extends this
result to dynamic weighted network design while making two notable advances.
First, our approach relies on discrete weight reallocation rather than
marginal adjustments, eliminating the need for differentiability
assumptions. This allows Proposition \ref{pro:weighted-KB2} to be extended
to any convex function of nodes' equilibrium efforts. Second, our
weight-reallocation operations exclude certain weighted NSGs from
optimality, providing the foundation for the second part of the proposition.

\bigskip \noindent\textbf{The Proof Sketch of Proposition \ref%
{pro:weighted-KB2} (ii)}. To establish the second part of Proposition \ref%
{pro:weighted-KB2}, we must identify which weighted NSGs that succeed a
quasi-complete (QC) graph maximize the sum of squared KB centralities.
Figure \ref{fig:weighted QC} illustrates three representative classes of
weighted NSGs formed by adding one unit of weight to a QC graph with 6 nodes
and 4 links. Dotted lines indicate weighted links ($0 < $ weight $< 1$),
while solid lines represent links with weight 1. The third typical class of
weighted NSGs is exactly the (unweighted) QC graph, unique up to
isomorphism. The second class includes the other unweighted NSG by setting $%
\alpha=0$, and any strictly convex combination of this unweighted NSG and
the QC graph. The first class is the rest weighted NSGs succeeding the QC
graph. For the first class of weighted NSGs, there always exists a node with
more than two strictly weighted links. In the example, node $1$ has two
strictly weighted links $(1,5)$ and $(1,6)$ and node $4$ has two strictly
weighted links $(4,2)$ and $(4,6)$. In the first step, by iteratively
adopting the weight reallocation operation proposed in Lemma \ref%
{lem:sun2023} (in the example, it is switching weights from $(1,5)$ to $%
(1,6) $ and weights from $(4,2)$ to $(4,6)$), we can show that the third
class of NSGs is strictly dominated by the union of the first and second
classes. In the second step, we extend Lemma \ref{lem:unweighted-QC} by
showing that the QC graph dominates any graph in the second class in
aggregate walks for any length $k\geq 2$. These two steps are formally
presented by Lemma \ref{lem:weighted-KB2} in the Appendix.

\begin{figure}[tph]
\begin{center}
\includegraphics[scale=0.6]{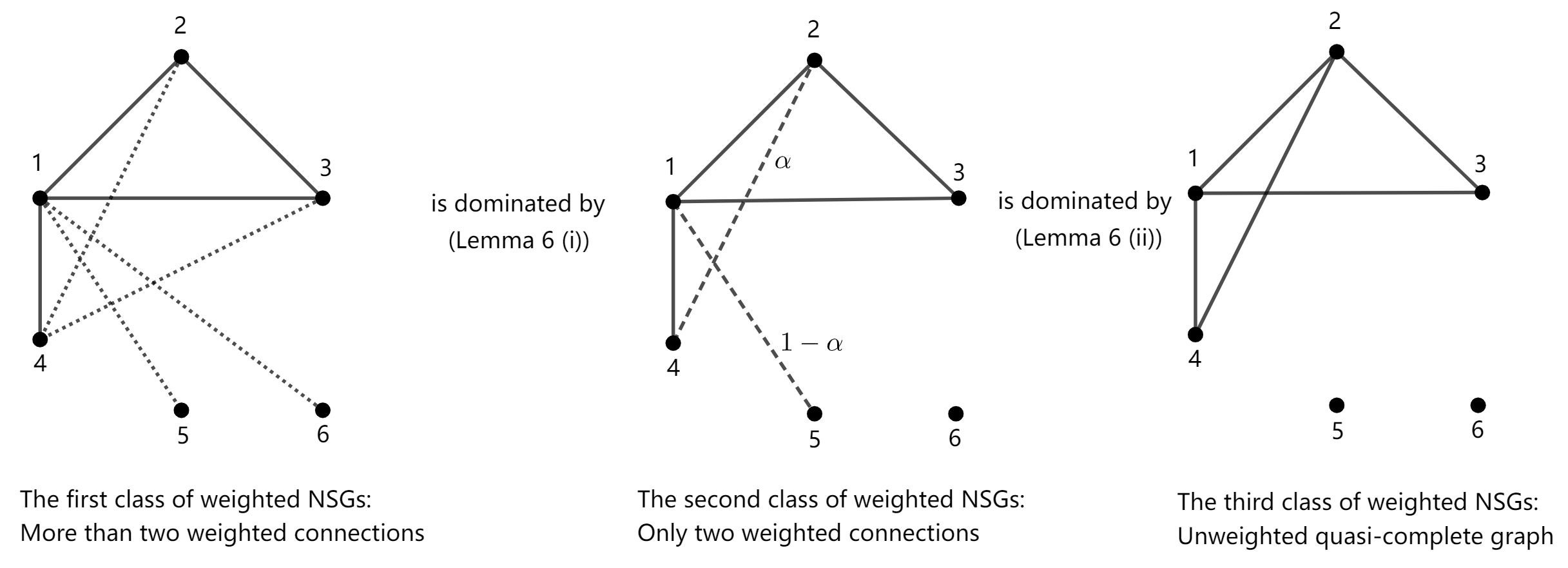}
\end{center}
\caption{The weighted NSGs obtained by adding one unit of weight.}
\label{fig:weighted QC}
\end{figure}

\section{Conclusion and Discussion}

\label{sec-conclusion}

This paper examines sequential network design within a general framework.
The social planner constructs a network over $T$ periods by adding one link
to the previously formed network in each period. The planner's preference
between network formation processes is determined by which process generates
networks with higher total numbers of walks of arbitrary lengths in each
period. This preference structure captures both the planner's varying degree
of farsightedness and the objective of maximizing diverse aggregate
centrality measures. Our analysis establishes that the optimal formation
process yields an NSG in each period. Specifically, when the planner is
myopic, the optimal strategy produces a QC graph in each period. We further
extend these results to weighted network design problems, where the planner
aims to maximize the stream of aggregate (or squared) Katz-Bonacich
centrality generated throughout the formation process.

This results can be (partially) extended to the case of heterogeneous nodes.
Specifically, consider the utility function of agent~$i$ given by 
\begin{equation*}
u_i(\mathbf{a}, \mathbf{G}) = \theta_i a_i - \frac{1}{2}a_i^2 + \phi \sum_{j
\in N} g_{ij} a_i a_j,
\end{equation*}
where the intrinsic marginal utilities $\boldsymbol{\theta} = (\theta_i)_{i
\in N}$ of agents are distinct. Then the structure of network $\mathbf{G}$
and intrinsic marginal utilities $\boldsymbol{\theta}$ jointly determine the
equilibrium, which is given by the weighted Katz-Bonacich centrality $%
\mathbf{a}^*(\mathbf{G}, \boldsymbol{\theta}) = (\mathbf{I} - \phi \mathbf{G}%
)^{-1} \boldsymbol{\theta}$. Assume the planner forms the network
dynamically to maximize the aggregate equilibrium activities. To analyze the
dynamic model of network formation with agent heterogeneity, we need to
modify Lemma~\ref{lem:wel-increase} as follows:

When $\theta_{i}\geq \theta_{j}$ and $\mathcal{T}_{j \rightarrow i}(\mathbf{G%
})$ is not isomorphic to $\mathbf{G}$, 
\begin{equation*}
\mathbf{1}^{\prime }\mathbf{G}^k \boldsymbol{\theta} < \mathbf{1}^{\prime }(%
\mathcal{T}_{j \rightarrow i}(\mathbf{G}))^k \boldsymbol{\theta} \text{ for
any integer } k \geq 2.
\end{equation*}

The proof is analogous to that of Lemma 1, with agent heterogeneity further amplifying the resulting inequalities. According to this result, the social planner tends to connect agents with high $\theta$ values first before connecting agents with divergent or polarized $\theta$ values. Consequently, the optimal strategy forms networks in which the neighbors of an agent with low $\theta$ are nested within those of an agent with high $\theta$. As a result, the formed network is also an NSG in each period, and Theorem 1  continues to holds. However, Theorem 2 does not hold under heterogeneity: if an agent's $\theta$ is substantially larger than others', then the optimal network formed in period $t = n - 1$, denoted $\mathbf{G}^*(n - 1)$, is a star network centered on the the high-$\theta$ agent.

Our developed framework is distinctive for its breadth, showing that
networks optimized for maximizing different walk-based centrality measures,
such as Katz-Bonacich centrality, diffusion centrality, and community
centrality, exhibit similar structural characteristics. As demonstrated in
Lemma \ref{lem:non-linear} in the Appendix, our primary finding (Theorem \ref%
{thm:unweighted-NSG}) applies even to situations where the planner's
immediate utility is derived from strategic complementary network games with
convex best response functions. Nevertheless, our approach has limitations
in contexts where walk-based measures inadequately reflect network value.
Exploring these challenges necessitates additional research and alternative
methodological approaches.

In our analysis, we have distinguished between two different NSGs based on
their spectral radius. The challenge of identifying the graph with the
maximum spectral radius was first presented by \cite{Brualdi1985} and has
remained unsolved in mathematics for over three decades. Our work provides
an initial contribution toward differentiating between NSGs (Lemma \ref%
{lem:unweighted-QC}), offering a foundation for further exploration in this
area.

\newpage

\section{Appendix A: Proofs}

\noindent\textbf{Proof of Lemma \ref{lem:wel-increase}.} We prove a stronger
lemma here, which can be used to show the robustness of Theorem \ref%
{thm:unweighted-NSG} when planner's instantaneous utility is micro-founded
by a network game with a convex best response function.

Consider a connected network $\mathbf{G}$ and two distinct nodes $i,j$ such
that $N_{j}\left( \mathbf{G}\right) \backslash \{i\}\subset N_{i}\left( 
\mathbf{G}\right) \backslash \{j\}$. Given an arbitrary set of nodes $%
L=\{l_{1},...,l_{k}\}\subseteq N\backslash \{i,j\}$ such that $L\cap
N_{i}\left( \mathbf{G}\right) =\emptyset $, denote $\mathbf{\hat{G}}:=%
\mathbf{G+}\underset{l\in L}{\sum }\mathbf{E}_{il}$ and $\mathbf{\bar{G}}:=%
\mathbf{G+}\underset{l\in L}{\sum }\mathbf{E}_{jl}$. For any non-constant function $\psi
\left( \cdot \right) :\mathbb{R}_{0}^{+}\rightarrow \mathbb{R}_{0}^{+}$
define the operater $\Lambda _{\mathbf{G},\psi }:\mathbb{R}^{n}\rightarrow 
\mathbb{R}^{n}$ by $\left( \Lambda _{\mathbf{G},\psi }\left( x\right)
\right) _{i}=\psi \left( \underset{k\in N_{i}\left( \mathbf{G}\right) }{\sum 
}x_{k}\right) $, $\forall i\in N$. Let $\mathbf{x}^{\left( m\right)
}=\Lambda _{\mathbf{\hat{G}},\psi }\left( \mathbf{x}^{\left( m-1\right)
}\right) $ and $\mathbf{y}^{\left( m\right) }=\Lambda _{\mathbf{\bar{G}}%
,\psi }\left( \mathbf{y}^{\left( m-1\right) }\right) $ denote the $m$-th
iterations starting from the unit vector under networks $\mathbf{\hat{G}}$
and $\mathbf{\bar{G}}$, respectively. Let $\mathbf{x}^{0
}>\mathbf{0}$ and $\mathbf{y}^{0
}>\mathbf{0}$.

\begin{lem}
\label{lem:non-linear}When $\psi \left( \cdot \right) $ is convex and $\psi
^{\prime }(\cdot )\in (0,1]$, we have 
\begin{equation*}
\sum_{k\in N}x_{k}^{\left( m\right) }>\sum_{k\in N}y_{k}^{\left( m\right) }%
\text{, }\forall m\geq 2\text{.}
\end{equation*}
\end{lem}

When $\psi \left( \cdot \right) $ is the identity function and $\mathbf{x}%
^{\left( 0\right) }=\mathbf{y}^{\left( 0\right) }=\mathbf{1}$, $%
x_{k}^{\left( m\right) }$ and $y_{k}^{\left( m\right) }$ count the total
number of walks of length $m$ starting from node $k$ in the network $\mathbf{%
\hat{G}}$ and $\mathbf{\bar{G}}$, respectively. Therefore, Lemma \ref%
{lem:non-linear} covers Lemma \ref{lem:wel-increase} by setting $\psi \left(
\cdot \right) $ as the identity function and viewing the original network $%
\mathbf{G}$ in the statement of Lemma \ref{lem:wel-increase} as $\mathbf{%
\bar{G}}$ in Lemma \ref{lem:non-linear}.

We use mathematical induction to show that for all $m\geq 0$, the following
four statements hold:

1. $x_{k}^{\left( m\right) }\geq y_{k}^{\left( m\right) }$, $\forall k\neq j$%
;

2. $x_{i}^{\left( m\right) }\geq y_{j}^{\left( m\right) }$;

3. $x_{i}^{\left( m\right) }+x_{j}^{\left( m\right) }\geq y_{i}^{\left(
m\right) }+y_{j}^{\left( m\right) }$;

4. $x_{k}^{\left( m\right) }$ and $y_{k}^{\left( m\right) }$ is strictly
increasing in $m$ for any $k$.

When $m=0$ and $1$, the four arguments trivially hold. Assume that these
four arguments hold for any $m\leq m^{\prime }$ where $m^{\prime }\geq 1$.
The forth argument holds straightforwardly by the definitions of $%
x_{k}^{\left( m^{\prime }\right) }$, $y_{k}^{\left( m^{\prime }\right) }$
and the inductive assumption since $\psi ^{\prime }(\cdot )>0$.

We first show that $x_{k}^{\left( m^{\prime }+1\right) }\geq y_{k}^{\left(
m^{\prime }+1\right) }$ for any $k\notin \left\{ i,j\right\} \cup L$. 
\begin{eqnarray*}
x_{k}^{\left( m^{\prime }+1\right) } &=&\psi ( \underset{k^{\prime }\notin
\left\{ i,j\right\} }{\sum }g_{kk^{\prime }}x_{k^{\prime }}^{\left(
m^{\prime }\right) }+g_{ki}x_{i}^{\left( m^{\prime }\right)
}+g_{kj}x_{j}^{\left( m^{\prime }\right) }) \geq \psi ( \underset{k^{\prime
}\notin \left\{ i,j\right\} }{\sum }g_{kk^{\prime }}y_{k^{\prime }}^{\left(
m^{\prime }\right) }+g_{ki}x_{i}^{\left( m^{\prime }\right)
}+g_{kj}x_{j}^{\left( m^{\prime }\right) }) \\
&\geq &\psi ( \underset{k^{\prime }\notin \left\{ i,j\right\} }{\sum }%
g_{kk^{\prime }}y_{k^{\prime }}^{\left( m^{\prime }\right)
}+g_{ki}y_{i}^{\left( m^{\prime }\right) }+g_{kj}y_{j}^{\left( m^{\prime
}\right) }) =y_{k}^{\left( m^{\prime }+1\right) }\text{.}
\end{eqnarray*}%
The first inequality follows from the fact that $x_{k}^{\left( m^{\prime
}\right) }\geq y_{k}^{\left( m^{\prime }\right) }$, $\forall k\notin \left\{
i,j\right\} $. The second inequality follows from the inductive assumption
and the fact that $g_{ki}\geq g_{kj}$. Note that, for some node $k$ such
that $g_{ki}=1$ and $g_{kj}=0$, we have $x_{k}^{\left( m^{\prime }+1\right)
}>y_{k}^{\left( m^{\prime }+1\right) }$.

Then, we are going to show that, for any $l\in L$, $x_{l}^{\left( m^{\prime
}+1\right) }\geq y_{l}^{\left( m^{\prime }+1\right) }$. The inequality holds
since%
\begin{equation*}
x_{l}^{\left( m^{\prime }+1\right) }=\psi ( \underset{k^{\prime }}{\sum }%
g_{lk^{\prime }}x_{k^{\prime }}^{\left( m^{\prime }\right) }+x_{i}^{\left(
m^{\prime }\right) }) \geq \psi ( \underset{k^{\prime }}{\sum }g_{lk^{\prime
}}y_{k^{\prime }}^{\left( m^{\prime }\right) }+y_{j}^{\left( m^{\prime
}\right) }) =y_{l}^{\left( m^{\prime }+1\right) }\text{.}
\end{equation*}

Third, we compare $x_{i}^{\left( m^{\prime }+1\right) }$ and $y_{i}^{\left(
m^{\prime }+1\right) }$. Note that 
\begin{equation*}
x_{i}^{\left( m^{\prime }+1\right) }=\psi (g_{ij}x_{j}^{\left( m^{\prime
}\right) }+\underset{k\neq j}{\sum }g_{ik}x_{k}^{\left( m^{\prime }\right) }+%
\underset{l\in L}{\sum }x_{l}^{\left( m^{\prime }\right) })\text{.}
\end{equation*}%
When $g_{ij}=0$, we have $x_{i}^{\left( m^{\prime }+1\right) }\geq
y_{i}^{\left( m^{\prime }+1\right) }$ since $x_{k}^{\left( m^{\prime
}\right) }\geq y_{k}^{\left( m^{\prime }\right) }>1$ (by the fact that $%
x_{k}^{\left( 0\right) }=y_{k}^{\left( 0\right) }=1$ and $x_{k}^{\left(
m\right) }$, $y_{k}^{\left( m\right) }$ increasing in $m$), $\forall k\neq j$%
. When $g_{ij}=1$, we have 
\begin{eqnarray*}
x_{i}^{\left( m^{\prime }+1\right) } &=&\psi (x_{j}^{\left( m^{\prime
}\right) }+\underset{k\neq j}{\sum }g_{ik}x_{k}^{\left( m^{\prime }\right) }+%
\underset{l\in L}{\sum }x_{l}^{\left( m^{\prime }\right) }) \\
&=&\psi (\psi (x_{i}^{\left( m^{\prime }-1\right) }+\underset{k\neq i}{\sum }%
g_{jk}x_{k}^{\left( m^{\prime }-1\right) })+\underset{k\neq j}{\sum }%
g_{ik}x_{k}^{\left( m^{\prime }\right) }+\underset{l\in L}{\sum }%
x_{l}^{\left( m^{\prime }\right) }) \\
&\geq &\psi (\psi (x_{i}^{\left( m^{\prime }-1\right) }+\underset{k\neq i}{%
\sum }g_{jk}x_{k}^{\left( m^{\prime }-1\right) })+\underset{k\neq j}{\sum }%
g_{ik}x_{k}^{\left( m^{\prime }\right) }+\underset{l\in L}{\sum }%
x_{l}^{\left( m^{\prime }-1\right) }) \\
&\geq &\psi (\psi (y_{i}^{\left( m^{\prime }-1\right) }+\underset{k\neq i}{%
\sum }g_{jk}y_{k}^{\left( m^{\prime }-1\right) })+\underset{k\neq j}{\sum }%
g_{ik}y_{k}^{\left( m^{\prime }\right) }+\underset{l\in L}{\sum }%
y_{l}^{\left( m^{\prime }-1\right) }) \\
&\geq &\psi (\psi (y_{i}^{\left( m^{\prime }-1\right) }+\underset{k\neq i}{%
\sum }g_{jk}y_{k}^{\left( m^{\prime }-1\right) }+\underset{l\in L}{\sum }%
y_{l}^{\left( m^{\prime }-1\right) })+\underset{k\neq j}{\sum }%
g_{ik}y_{k}^{\left( m^{\prime }\right) })=y_{i}^{\left( m^{\prime }+1\right)
}\text{.}
\end{eqnarray*}%
The first inequality follows from the fact that $x_{k}^{\left( m\right) }$
is increasing in $m$. The second inequality follows from $x_{k}^{\left(
m\right) }\geq y_{k}^{\left( m\right) }$, $\forall k\neq j$ and $m\leq
m^{\prime }$. The third inequality follows from the fact that $\psi ^{\prime
}\left( \cdot \right) \leq 1$.

We further show that $x_{i}^{\left( m^{\prime }+1\right) }\geq y_{j}^{\left(
m^{\prime }+1\right) }$. The argument trivially holds when $g_{ij}=0$, and
we focus on the case of $g_{ij}=1$.

Decomposing $x_{i}^{\left( m^{\prime }+1\right) }$ and using the inductive
assumptions, 
\begin{eqnarray*}
x_{i}^{\left( m^{\prime }+1\right) } &=&\psi (\psi (x_{i}^{\left( m^{\prime
}-1\right) }+\underset{k\neq i}{\sum }g_{jk}x_{k}^{\left( m^{\prime
}-1\right) })+\underset{k\neq j}{\sum }g_{ik}x_{k}^{\left( m^{\prime
}\right) }+\underset{l\in L}{\sum }x_{l}^{\left( m^{\prime }\right) }) \\
&\geq &\psi (\psi (y_{j}^{\left( m^{\prime }-1\right) }+\underset{k\neq i}{%
\sum }g_{jk}x_{k}^{\left( m^{\prime }-1\right) })+\underset{k\neq j}{\sum }%
g_{ik}x_{k}^{\left( m^{\prime }\right) }+\underset{l\in L}{\sum }%
x_{l}^{\left( m^{\prime }\right) }) \\
&\geq &\psi (\psi (y_{j}^{\left( m^{\prime }-1\right) }+\underset{k\neq i}{%
\sum }g_{jk}y_{k}^{\left( m^{\prime }-1\right) })+\underset{k\neq j}{\sum }%
g_{ik}y_{k}^{\left( m^{\prime }\right) }+\underset{l\in L}{\sum }%
y_{l}^{\left( m^{\prime }\right) }) \\
&=&\psi (\psi (y_{j}^{\left( m^{\prime }-1\right) }+\underset{k\neq i}{\sum }%
g_{jk}y_{k}^{\left( m^{\prime }-1\right) })+\underset{k\neq j,k\in
N_{i}\left( \mathbf{G}\right) \backslash N_{j}\left( \mathbf{G}\right) }{%
\sum }y_{k}^{\left( m^{\prime }\right) }+\underset{k\neq i}{\sum }%
g_{jk}y_{k}^{\left( m^{\prime }\right) }+\underset{l\in L}{\sum }%
y_{l}^{\left( m^{\prime }\right) })\text{.}
\end{eqnarray*}%
Moreover, since $y_{k}^{\left( m\right) }$ is increasing in $m$, we can get 
\begin{eqnarray*}
x_{i}^{\left( m^{\prime }+1\right) } &\geq &\psi (\psi (y_{j}^{\left(
m^{\prime }-1\right) }+\underset{k\neq i}{\sum }g_{jk}y_{k}^{\left(
m^{\prime }-1\right) })+\underset{k\neq j,k\in N_{i}\left( \mathbf{G}\right)
\backslash N_{j}\left( \mathbf{G}\right) }{\sum }y_{k}^{\left( m^{\prime
}-1\right) }+\underset{k\neq i}{\sum }g_{jk}y_{k}^{\left( m^{\prime }\right)
}+\underset{l\in L}{\sum }y_{l}^{\left( m^{\prime }\right) }) \\
&\geq &\psi (\psi (y_{j}^{\left( m^{\prime }-1\right) }+\underset{k\neq i}{%
\sum }g_{jk}y_{k}^{\left( m^{\prime }-1\right) }+\underset{k\neq j,k\in
N_{i}\left( \mathbf{G}\right) \backslash N_{j}\left( \mathbf{G}\right) }{%
\sum }y_{k}^{\left( m^{\prime }-1\right) })+\underset{k\neq i}{\sum }%
g_{jk}y_{k}^{\left( m^{\prime }\right) }+\underset{l\in L}{\sum }%
y_{l}^{\left( m^{\prime }\right) }) \\
&=&\psi (\psi (y_{j}^{\left( m^{\prime }-1\right) }+\underset{k\neq j}{\sum }%
g_{ik}y_{k}^{\left( m^{\prime }-1\right) })+\underset{k\neq i}{\sum }%
g_{jk}y_{k}^{\left( m^{\prime }\right) }+\underset{l\in L}{\sum }%
y_{l}^{\left( m^{\prime }\right) })=y_{j}^{\left( m^{\prime }+1\right) }%
\text{,}
\end{eqnarray*}%
where the last inequality comes from $\psi ^{\prime }\left( \cdot \right)
\leq 1$.

Finally, we show that $x_i^{(m^{\prime }+1)} + x_j^{(m^{\prime }+1)} \geq
y_i^{(m^{\prime }+1)} + y_j^{(m^{\prime }+1)}$. Note that since $\psi(\cdot)$
is convex and increasing, for any four real numbers $a, b, c, d$, we have $%
\psi(a) + \psi(b) \geq \psi(c) + \psi(d)$ if $a + b \geq c + d$ and $%
\max\{a, b\} \geq \max\{c, d\}$. Therefore, 
\begin{align}
&x_i^{(m^{\prime }+1)} + x_j^{(m^{\prime }+1)} \geq y_i^{(m^{\prime }+1)} +
y_j^{(m^{\prime }+1)}  \notag \\
\Leftrightarrow \quad &\psi\left(\sum_k g_{ik} x_k^{(m^{\prime })} + \sum_{l
\in L} x_l^{(m^{\prime })}\right) + \psi\left(\sum_k g_{jk} x_k^{(m^{\prime
})}\right)  \notag \\
&\qquad\qquad \geq \psi\left(\sum_k g_{ik} y_k^{(m^{\prime })}\right) +
\psi\left(\sum_k g_{jk} y_k^{(m^{\prime })} + \sum_{l \in L} y_l^{(m^{\prime
})}\right)  \notag \\
\Leftarrow \quad &\sum_k g_{ik} x_k^{(m^{\prime })} + \sum_{l \in L}
x_l^{(m^{\prime })} + \sum_k g_{jk} x_k^{(m^{\prime })} \geq \sum_k g_{jk}
y_k^{(m^{\prime })} + \sum_{l \in L} y_l^{(m^{\prime })} + \sum_k g_{ik}
y_k^{(m^{\prime })}  \label{eq:thm11} \\
\text{and} \quad &\sum_k g_{ik} x_k^{(m^{\prime })} + \sum_{l \in L}
x_l^{(m^{\prime })} \geq \max\left\{\sum_k g_{jk} y_k^{(m^{\prime })} +
\sum_{l \in L} y_l^{(m^{\prime })}, \sum_k g_{ik} y_k^{(m^{\prime
})}\right\}.  \label{eq:thm12}
\end{align}
Equation \eqref{eq:thm11} holds by the inductive assumption that $%
x_i^{(m^{\prime })} + x_j^{(m^{\prime })} \geq y_i^{(m^{\prime })} +
y_j^{(m^{\prime })}$ and $x_k^{(m^{\prime })} \geq y_k^{(m^{\prime })}$ for
all $k \neq j$.

To show equation \eqref{eq:thm12}, note that we have shown that $%
x_{i}^{(m^{\prime }+1)}\geq max\{y_{j}^{(m^{\prime }+1)},y_{i}^{(m^{\prime
}+1)}\}$. By the monotonicity of $\phi(\cdot)$, we have the equality.

\bigskip

The strictness when $m\geq 2$ comes from the fact that 
\begin{equation*}
x_{i}^{m^{\prime }}>y_{j}^{m^{\prime }}\implies x_{l}^{m^{\prime
}+1}>y_{l}^{m^{\prime }+1}\implies x_{i}^{m^{\prime }+1}>y_{j}^{m^{\prime
}+1}
\end{equation*}%
by the decomposition of $x_{l}^{m^{\prime }+1}$, $y_{l}^{m^{\prime }+1}$ and
inductive assumptions. Moreover, when $L\neq \emptyset $, the inequality $%
x_{i}^{1}>y_{j}^{1}$ holds for the degrees of nodes $i$ and $j$. Combined
with the other weak inequalities, we have $\sum_{k}x_{k}^{m}>%
\sum_{k}y_{k}^{m}$ when $m>2$.

The result that $\lambda_{\max}(\mathbf{G}) < \lambda_{\max}(\mathcal{T}_{j
\to i}(\mathbf{G}))$ follows directly from Theorem 1 in \cite{Wu2005}.
Specifically, Theorem 1 in \cite{Wu2005} states that in a network $\mathbf{G}
$ where node $i$'s eigenvector centrality is weakly larger than that of node 
$j$, and there exist nodes $\{v_1, \ldots, v_s\}$ that are connected to $j$
but not to $i$, the following holds. Let $\mathbf{G}^*$ be the network
obtained from $\mathbf{G}$ by deleting links $(j, v_k)$ and adding links $%
(i, v_k)$ for $1 \leq k \leq s$. Then $\lambda_{\max}(\mathbf{G}) <
\lambda_{\max}(\mathbf{G}^*)$.

In Lemma~\ref{lem:wel-increase}, since $\mathbf{G}$ and $\mathcal{T}_{j \to
i}(\mathbf{G})$ are not isomorphic, the set $L = \{l \in N \setminus \{i,j\}
: g_{il} = 0 \text{ and } g_{jl} = 1\}$ is nonempty. Therefore, $%
\lambda_{\max}(\mathbf{G}) < \lambda_{\max}(\mathcal{T}_{j \to i}(\mathbf{G}%
)) = \lambda_{\max}(\mathcal{T}_{i \to j}(\mathbf{G}))$ by Theorem 1 in \cite%
{Wu2005}. Note that Theorem 1 in \cite{Wu2005} requires that node $i$'s
eigenvector centrality be weakly higher than that of node $j$. Here, we do
not need to impose such a requirement since $\mathcal{T}_{j \to i}(\mathbf{G}%
) \cong \mathcal{T}_{i \to j}(\mathbf{G})$. When $i$'s eigenvector
centrality is larger, we apply Theorem 1 in \cite{Wu2005} to $\mathcal{T}_{j
\to i}(\mathbf{G})$; if $j$'s eigenvector centrality is larger, we apply it
to $\mathcal{T}_{i \to j}(\mathbf{G})$. \hfill $\Box $

\bigskip

\noindent\textbf{Proof of Theorem~\ref{thm:unweighted-NSG}.} Suppose by
contradiction that there is an optimal path $\mathbf{s}=(\mathbf{G}%
(t))_{t=1}^T$ with some $t$ for which $\mathbf{G}(t)\notin \mathcal{NSG}(t)$%
. Let $t^{\prime }$ be the first time at which $\mathbf{G}(t^{\prime
})\notin \mathcal{NSG}(t^{\prime })$ along $\mathbf{s}$ (so $\mathbf{G}(t)\in%
\mathcal{NSG}(t)$ for all $t<t^{\prime }$). Since $\mathbf{G}(t^{\prime })$
fails to be a nested split graph, there exist distinct $i,j\in N$ such that
the neighborhoods were nested up to $t^{\prime }$ (i.e., $N_j(\mathbf{G}%
(t))\setminus\{i\}\subseteq N_i(\mathbf{G}(t))\setminus\{j\}$ for all $%
t<t^{\prime }$), but at time $t^{\prime }$ nesting fails: 
\begin{equation*}
N_j(\mathbf{G}(t^{\prime }))\setminus\{i\}\not\subseteq N_i(\mathbf{G}%
(t^{\prime }))\setminus\{j\}.
\end{equation*}
Equivalently, there exists $\ell\notin\{i,j\}$ with $g_{j\ell}(t^{\prime
})=1 $ and $g_{i\ell}(t^{\prime })=0$.

\smallskip We construct a new feasible path $\hat{\mathbf{s}}=(\hat{\mathbf{G%
}}(t))_{t=1}^T$ that coincides with $\mathbf{s}$ up to $t^{\prime }-1$ and,
from $t^{\prime }$ onward, simulates the same sequence of link additions
except that whenever $\mathbf{s}$ adds a link involving $j$ and some $%
\ell\notin\{i,j\}$, the modified path adds that link to $i$ if available;
otherwise it follows the original addition. Formally, set $\hat{\mathbf{G}}%
(t)=\mathbf{G}(t)$ for $t<t^{\prime }$, and for $t\ge t^{\prime }$ define
recursively:

\begin{enumerate}
\item If $\mathbf{G}(t+1)=\mathbf{G}(t)+\mathbf{E}_{j\ell}$ with $%
\ell\notin\{i,j\}$, then 
\begin{equation*}
\hat{\mathbf{G}}(t+1)= 
\begin{cases}
\hat{\mathbf{G}}(t)+\mathbf{E}_{i\ell}, & \text{if }\hat{g}_{i\ell}(t)=0, \\ 
\hat{\mathbf{G}}(t)+\mathbf{E}_{j\ell}, & \text{if }\hat{g}_{i\ell}(t)=1.%
\end{cases}%
\end{equation*}

\item If $\mathbf{G}(t+1)=\mathbf{G}(t)+\mathbf{E}_{i\ell}$ with $%
\ell\notin\{i,j\}$, then 
\begin{equation*}
\hat{\mathbf{G}}(t+1)= 
\begin{cases}
\hat{\mathbf{G}}(t)+\mathbf{E}_{i\ell}, & \text{if }\hat{g}_{i\ell}(t)=0, \\ 
\hat{\mathbf{G}}(t)+\mathbf{E}_{j\ell}, & \text{if }\hat{g}_{i\ell}(t)=1.%
\end{cases}%
\end{equation*}

\item If $\mathbf{G}(t+1)=\mathbf{G}(t)+\mathbf{E}_{\ell k}$ with $%
\{\ell,k\}\cap\{i,j\}=\emptyset$, or $(\ell,k)=(i,j)$, then set $\hat{%
\mathbf{G}}(t+1)=\hat{\mathbf{G}}(t)+\mathbf{E}_{\ell k}$.
\end{enumerate}

\smallskip We show by induction that at each step the required addition is
feasible, i.e., when we prescribe adding $\mathbf{E}_{ab}$ we indeed have $%
\hat{g}_{ab}(t)=0$.

- For case (3), feasibility is immediate since we mirror the original step
on a pair $(\ell,k)$ with $\{\ell,k\}\cap\{i,j\}=\emptyset$ (or $(i,j)$),
and outside $\{i,j\}$ we never alter adjacency relative to $\mathbf{G}$
except by copying the same addition, so $\hat{g}_{\ell k}(t)=g_{\ell k}(t)=0$%
.

- For case (1), suppose $\mathbf{G}(t+1)=\mathbf{G}(t)+\mathbf{E}_{j\ell}$
with $\ell\notin\{i,j\}$. If $\hat{g}_{i\ell}(t)=0$, we add $(i,\ell)$;
feasibility is clear. If $\hat{g}_{i\ell}(t)=1$, we instead add $(j,\ell)$.
Assume toward contradiction that $\hat{g}_{j\ell}(t)=1$ already. Since $%
g_{j\ell}(t)=0$ (the original adds it at $t+1$), there must exist an earlier
time $\tau\in\{t^{\prime },\dots,t\}$ at which $\hat{\mathbf{G}}$ added $%
(j,\ell)$. Tracing back the rules, an addition of $(j,\ell)$ can only occur
in case (2) when the original added $(i,\ell)$ and $\hat{g}_{i\ell}=1$, or
in case (1) when the original added $(j,\ell)$ but $\hat{g}_{i\ell}=1$. In
either subcase, one finds a strictly earlier time at which $(i,\ell)$ had
been added in $\hat{\mathbf{G}}$ while it was absent in $\mathbf{G}$, which
(by the minimality of $t^{\prime }$ and the rule (1a)) forces an earlier
original addition of $(j,\ell)$, contradicting the fact that
the first original addition of $(j,\ell)$ occurs at $t+1$. Hence $\hat{g}%
_{j\ell}(t)=0$ and the step is feasible.

- For case (2), suppose $\mathbf{G}(t+1)=\mathbf{G}(t)+\mathbf{E}_{i\ell}$
with $\ell\notin\{i,j\}$. If $\hat{g}_{i\ell}(t)=0$, feasibility is trivial.
If $\hat{g}_{i\ell}(t)=1$, we add $(j,\ell)$. If this were infeasible, we
would have $\hat{g}_{j\ell}(t)=1$. Since $g_{i\ell}(t)=0$ (the original adds
it at $t+1$) and $\hat{g}_{i\ell}(t)=1$, there must have been a time $%
\tau\in\{t^{\prime },\dots,t\}$ when $\hat{\mathbf{G}}$ added $(i,\ell)$ via
rule (1a), which requires that the original at that $\tau$ added $(j,\ell)$.
But then $g_{j\ell}(\tau-1)=0$, and by the same logic as above we cannot
have $\hat{g}_{j\ell}(t)=1$ before $t+1$ without contradicting that the
original first adds $(i,\ell)$ at $t+1$. Hence $\hat{g}_{j\ell}(t)=0$ and
the step is feasible.

Thus every step adds a previously absent link, so $\hat{\mathbf{s}}\in S$ is
a valid successive path. Moreover, $\hat{\mathbf{G}}(t^{\prime })$ is
obtained from $\mathbf{G}(t^{\prime })$ by shifting some neighbors of $j$ to 
$i$ while keeping the total number of links at $t^{\prime }$, i.e., 
\begin{equation*}
\hat{\mathbf{G}}(t^{\prime }) \;=\; \mathcal{T}_{j\to i}\bigl(\,\mathbf{G}%
(t^{\prime })\,\bigr),
\end{equation*}
with $L:=\{\ell\notin\{i,j\}: g_{j\ell}(t^{\prime })=1,\,
g_{i\ell}(t^{\prime })=0\}\neq\emptyset$ by the choice of $t^{\prime }$ and $%
(i,j)$.

By construction, for each $t\ge t^{\prime }$ the set 
\begin{equation*}
L(t):=\{\ell\notin\{i,j\}: \hat{g}_{i\ell}(t)>g_{i\ell}(t)\}
\end{equation*}
coincides with $\{\ell\notin\{i,j\}: g_{j\ell}(t)>\hat{g}_{j\ell}(t)\}$, and
we can write 
\begin{equation*}
\hat{\mathbf{G}}(t) \;=\; \mathcal{T}_{j\to i}\bigl(\,\mathbf{G}(t)\,\bigr),
\quad\text{with } L(t)=\{\ell\notin\{i,j\}: g_{i\ell}(t)<g_{j\ell}(t)\}.
\end{equation*}
Indeed, whenever the original process creates a strict advantage of $j$ over 
$i$ at some neighbor $\ell$, rule (1a) immediately assigns $(i,\ell)$ in $%
\hat{\mathbf{G}}$; conversely, whenever $\hat{g}_{i\ell}(t)>g_{i\ell}(t)$,
that link must have been added by (1a) at the first time the original added $%
(j,\ell)$. Therefore, $L(t)$ is exactly the set of neighbors of $j$ missing
from $i$ at time $t$ in $\mathbf{G}(t)$.

Since $L(t)\neq\emptyset$ at $t=t^{\prime }$ and $\hat{\mathbf{G}}(t)$ is
not isomorphic to $\mathbf{G}(t)$ when a reallocation occurs, Lemma~\ref%
{lem:wel-increase} applies period by period. By Assumption~\ref%
{ass:preference}, this implies $u\bigl(\hat{\mathbf{G}}(t)\bigr)\geq u\bigl(%
\mathbf{G}(t)\bigr)$ for all $t\ge t^{\prime }$.

Finally, starting from any optimal path, we can repeatedly apply the above
modification at the first non-NSG period to obtain an optimal path $\mathbf{s%
}^*=(\mathbf{G}^*(t))_{t=1}^T$ such that $\mathbf{G}^*(t)\in\mathcal{NSG}(t)$
for all $t$ (the process terminates in finitely many steps since $S$ is
finite). This proves both parts of the theorem. \qed

\bigskip

\noindent\textbf{Proof of Proposition \ref{prop-alpha} and Corollary \ref%
{cor-TE2016}.} Both results follow directly from Theorem \ref%
{thm:unweighted-NSG}.\hfill$\Box$

\bigskip \noindent\textbf{Proof of Lemmas~\ref{lem:unweighted-QC} and~\ref%
{lem:weighted-KB2}(ii).} We prove Lemma~\ref{lem:weighted-KB2}(ii); Lemma~%
\ref{lem:unweighted-QC} follows as a special case.

Consider an unweighted QC graph $\mathbf{G}$ with $t$ links containing a
clique on $p\ge 2$ nodes. We first show that there are at most two
unweighted NSGs that succeed $\mathbf{G}$.

When $t=\frac{p(p-1)}{2}$ or $t=\frac{p(p+1)}{2}$, $\mathbf{G}$ consists of
a $p$-clique and isolated nodes. Then $\mathbb{S}(\mathbf{G})$ is a
singleton: the unique successor is obtained by adding one link between a
clique node and an isolated node; this graph is QC and, in particular, an
NSG.

Now suppose $\bar p < t < \frac{p(p+1)}{2}$, where $\bar p:=\frac{p(p-1)}{2}$
is the number of links in the clique. Thus the first $p$ nodes form a
clique; node $p+1$ connects to the first $t-\bar p$ nodes; and nodes $%
p+2,\dots,n$ (if any) are isolated. Classify nodes by (weakly) decreasing
degree into six classes:

\textbf{Class 1:} node $1$ (e.g., node 1 in $\mathbf{G}$ in Figure~\ref%
{fig:QCproof});

\textbf{Class 2:} nodes $2$ to $t - \bar{p}$; we use $i$ to represent Class
2 nodes, and the total number of nodes in Class 2 is denoted by $\beta = t - 
\bar{p} - 1$ (e.g., node 2 in $\mathbf{G}$ in Figure \ref{fig:QCproof});

\textbf{Class 3:} node $t - \bar{p} + 1$ (e.g., node 3 in $\mathbf{G}$ in
Figure \ref{fig:QCproof});

\textbf{Class 4:} nodes $t - \bar{p} + 2$ to $p$; we use $j$ to represent
Class 4 nodes, and the total number of nodes in Class 4 is denoted by $%
\gamma = p - 2 - \beta$ (no class 4 nodes in the example of $\mathbf{G}$ in
Figure \ref{fig:QCproof});

\textbf{Class 5:} node $p + 1$ (e.g., node 4 in $\mathbf{G}$ in Figure \ref%
{fig:QCproof}); and

\textbf{Class 6:} node $p + 2$ (e.g., node 5 in $\mathbf{G}$ in Figure \ref%
{fig:QCproof}), which is isolated.

Nodes $p+3,\dots,n$ are also isolated and play no role, so we omit them.

Let $D_{0},\dots,D_{m}$ be the degree partition of $\mathbf{G}$ from low to high. Then $\mathbf{G}$ is an (unweighted) NSG if and only if for any two nodes $i\in D_{x}$ and $j\in D_{y}$, we have $g_{ij}=1$ if and only if $x+y>m$ (see Theorem 1.2.4, point 6 in \cite{N.V.R.Mahadev1995}). In particular, adding a single edge preserves the NSG property only when it connects two
nodes whose degree pair lies on the Pareto frontier of degree sums. Among
all potential successors of $\mathbf{G}$, the only such edges are (i)
between Class~3 and Class~5 (nodes $t-\bar{p}+1$ and $p+1$), and (ii) between
Class~1 and Class~6 (nodes $1$ and $p+2$). Any other added edge violates
nestedness.

Therefore, there are exactly two unweighted NSG successors of $\mathbf{G}$: 
\begin{equation*}
\mathbb{S}(\mathbf{G}) \cap \mathcal{NSG} = \big\{\, \mathbf{G} + \mathbf{E}%
_{\,t-\bar p+1,\;p+1}\,,\ \ \mathbf{G} + \mathbf{E}_{\,1,\;p+2}\,\big\}.
\end{equation*}

For $\alpha\in[0,1]$, define the matrix $\mathbf{E}(\alpha)$ by 
\begin{equation*}
\mathbf{E}_{\,t-\bar p+1,\;p+1}(\alpha)=\mathbf{E}_{\,p+1,\;t-\bar
p+1}(\alpha)=\alpha,\qquad \mathbf{E}_{\,1,\;p+2}(\alpha)=\mathbf{E}%
_{\,p+2,\;1}(\alpha)=1-\alpha,
\end{equation*}
with all other entries zero. Then 
\begin{equation*}
\mathrm{conv}\big(\mathbb{S}(\mathbf{G})\cap \mathcal{NSG}\big) =\Big\{\, 
\mathbf{G}[\alpha] := \mathbf{G} + \mathbf{E}(\alpha) \;:\; \alpha\in[0,1] \,%
\Big\}.
\end{equation*}
This characterizes the convex hull of unweighted NSG successors and
completes the structural part used in Lemma~\ref{lem:weighted-KB2}(ii).

Note that, $\mathbf{G}\left[1\right] $ is quasi-complete and $\mathbf{G}%
\left[ 0\right] $ is another unweighted NSG succeeds $\mathbf{G}$.
Therefore, the proof of Lemma \ref{lem:weighted-KB2} (ii) covers that of
Lemma \ref{lem:unweighted-QC} which compares unweighted NSGs $\mathbf{G}%
\left[1\right] $ and $\mathbf{G}\left[ 0\right] $.

Define vectors $\mathbf{x}^{k}=(\mathbf{G}[1])^{k}\mathbf{1}$ and $\mathbf{y}%
^{k}=(\mathbf{G}[\alpha])^{k}\mathbf{1}$ for any $k$. We list all
decompositions that will be used:

\begin{align*}
&x_{i}^{m+1} = x_{1}^{m} + x_{t-\bar{p}+1}^{m} + x_{p+1}^{m} +
(\beta-1)x_{i}^{m} + \gamma x_{j}^{m}, & &x_{j}^{m+1} = x_{1}^{m} + x_{t-%
\bar{p}+1}^{m} + \beta x_{i}^{m} + (\gamma-1)x_{j}^{m}, \\
&x_{1}^{m+1} = x_{t-\bar{p}+1}^{m} + x_{p+1}^{m} + \beta x_{i}^{m} + \gamma
x_{j}^{m}, & &x_{t-\bar{p}+1}^{m+1} = x_{1}^{m} + x_{p+1}^{m} + \beta
x_{i}^{m} + \gamma x_{j}^{m}, \\
&x_{p+1}^{m+1} = x_{1}^{m} + x_{t-\bar{p}+1}^{m} + \beta x_{i}^{m}, &  \\
&y_{i}^{m+1} = y_{1}^{m} + y_{t-\bar{p}+1}^{m} + y_{p+1}^{m} +
(\beta-1)y_{i}^{m} + \gamma y_{j}^{m}, & &y_{j}^{m+1} = y_{1}^{m} + y_{t-%
\bar{p}+1}^{m} + \beta y_{i}^{m} + (\gamma-1)y_{j}^{m}, \\
&y_{1}^{m+1} = y_{t-\bar{p}+1}^{m} + y_{p+1}^{m} + (1-\alpha)y_{p+2}^{m} +
\beta y_{i}^{m} + \gamma y_{j}^{m}, & &y_{t-\bar{p}+1}^{m+1} = y_{1}^{m} +
\alpha y_{p+1}^{m} + \beta y_{i}^{m} + \gamma y_{j}^{m}, \\
&y_{p+1}^{m+1} = y_{1}^{m} + \alpha y_{t-\bar{p}+1}^{m} + \beta y_{i}^{m}, & 
&y_{p+2}^{m+1} = (1-\alpha)y_{1}^{m}
\end{align*}

We use mathematical induction to prove the following three claims. It is
easy to verify that these claims hold when $k=0,1$. Therefore, we assume the
claims hold for any $k\geq m$ and show that the statements hold for $m+1$.

\noindent\textbf{Claim 1.} For any $k$, $x_{1}^{k}+x_{t-\bar{p}+1}^{k}\geq
y_{1}^{k}+y_{t-\bar{p}+1}^{k}$, $x_{i}^{k}\geq y_{i}^{k}$ and $%
(\beta-1)x_{i}^{k}+\gamma x_{j}^{k}\geq (\beta-1)y_{i}^{k}+\gamma y_{j}^{k}$.

\begin{proof}
    
Note that the first two inequalities in Claim 1 imply 
\begin{equation}
x_{p+1}^{m+1}\geq y_{p+1}^{m+1}  \label{eq:6}
\end{equation}
since $x_{p+1}^{m+1}=x_{1}^{m}+x_{t-\bar{p}+1}^{m}+\beta x_{i}^{m}$ and $%
y_{p+1}^{m+1}=y_{1}^{m}+\alpha y_{t-p+1}^{m}+\beta y_{i}^{m}$ where $\alpha
\in [0,1]$. Moreover, 
\begin{eqnarray*}
x_{i}^{m+1} &=&x_{1}^{m}+x_{t-\bar{p}+1}^{m}+x_{p+1}^{m}+(\beta-1)x_{i}^{m}+%
\gamma x_{j}^{m} \\
&\geq &y_{1}^{m}+y_{t-\bar{p}+1}^{m}+y_{p+1}^{m}+(\beta-1)y_{i}^{m}+\gamma
y_{j}^{m}=y_{i}^{m+1}.
\end{eqnarray*}

We then show the last inequality in Claim 1. Decomposing both sides, we have 
\begin{eqnarray*}
(\beta -1)x_{i}^{m+1}+\gamma x_{j}^{m+1} 
&=&(\beta +\gamma -1)(x_{1}^{m}+x_{t-\bar{p}+1}^{m})+(\beta -1)x_{p+1}^{m} \\
&&+(\beta -1)^{2}x_{i}^{m}+\gamma \beta x_{i}^{m}+(\beta +\gamma -2)\gamma x_{j}^{m} \\
(\beta -1)y_{i}^{m+1}+\gamma y_{j}^{m+1} 
&=&(\beta +\gamma -1)(y_{1}^{m}+y_{t-\bar{p}+1}^{m})+(\beta -1)y_{p+1}^{m} \\
&&+(\beta -1)^{2}y_{i}^{m}+\gamma \beta y_{i}^{m}+(\beta +\gamma -2)\gamma y_{j}^{m}.
\end{eqnarray*}
Moreover, since 
\begin{eqnarray*}
&&((\beta -1)^{2}+\gamma \beta )x_{i}^{m}+(\beta +\gamma -2)\gamma
x_{j}^{m}+(1-\beta -\gamma )x_{i}^{m} \\
&=&(\beta +\gamma -2)(\beta -1)x_{i}^{m}+(\beta +\gamma -2)\gamma x_{j}^{m},
\end{eqnarray*}%
by the inductive assumption that $(\beta -1)x_{i}^{m}+\gamma x_{j}^{m}\geq
(\beta -1)y_{i}^{m}+\gamma y_{j}^{m}$, it is sufficient to show that 
\begin{equation*}
(\beta +\gamma -1)(x_{1}^{m}+x_{t-\bar{p}+1}^{m}-x_{i}^{m})+(\beta
-1)x_{p+1}^{m}\geq (\beta +\gamma -1)(y_{1}^{m}+y_{t-\bar{p}%
+1}^{m}-y_{i}^{m})+(\beta -1)y_{p+1}^{m}.
\end{equation*}%
Further decomposing both sides, we have 
\begin{align*}
\text{LHS}& =(\beta +\gamma -1)(x_{i}^{m-1}+x_{p+1}^{m-1}+\beta
x_{i}^{m-1}+\gamma x_{j}^{m-1})+(\beta -1)(x_{1}^{m-1}+x_{t-\bar{p}%
+1}^{m-1}+\beta x_{i}^{m-1}) \\
\text{RHS}& =(\beta +\gamma -1)(y_{i}^{m-1}+(1-\alpha )y_{p+2}^{m-1}+\alpha
y_{p+1}^{m-1}+\beta y_{i}^{m-1}+\gamma y_{j}^{m-1})+(\beta
-1)(y_{1}^{m-1}+\alpha y_{t-\bar{p}+1}^{m-1}+\beta y_{i}^{m-1})
\end{align*}%
We can show that $LHS\geq RHS$ by the inductive assumptions and the fact
that $y_{p+1}^{k}\geq y_{p+2}^{k}$ for any $k$.

Now, we prove $x_{1}^{m+1}+x_{t-\bar{p}+1}^{m+1}\geq y_{1}^{m+1}+y_{t-\bar{p}%
+1}^{m+1}$. 
\begin{eqnarray*}
x_{1}^{m+1}+x_{t-\bar{p}+1}^{m+1} &=&2x_{1}^{m}+x_{t-\bar{p}%
+1}^{m}+2x_{p+1}^{m}+x_{i}^{m}+(2\beta -2)x_{i}^{m}+2\gamma x_{j}^{m} \\
y_{1}^{m+1}+y_{t-\bar{p}+1}^{m+1} &=&y_{1}^{m}+y_{t-\bar{p}+1}^{m}+(1+\alpha
)y_{p+1}^{m}+(1-\alpha )y_{p+2}^{m}+2\beta y_{i}^{m}+2\gamma y_{j}^{m} \\
&\geq &y_{1}^{m}+y_{t-\bar{p}+1}^{m}+2y_{p+1}^{m}+2y_{i}^{m}+(2\beta
-2)y_{i}^{m}+2\gamma y_{j}^{m},
\end{eqnarray*}%
where in the last equality we use the fact that $y_{p+1}^{k}\geq y_{p+2}^{k}$
for any $k$. By the inductive assumptions, $x_{1}^{m+1}+x_{t-\bar{p}%
+1}^{m+1}\geq y_{1}^{m+1}+y_{t-\bar{p}+1}^{m+1}$ whenever $%
x_{1}^{m}-x_{i}^{m}\geq 0$. This holds since node $1$ and node $i$ have the
same set of neighbors and thus $x_{1}^{k}=x_{i}^{k}$ for any $k$ (nodes 1
and 2 in $\mathbf{G}$ in Figure \ref{fig:QCproof}) in the unweighted network 
$\mathbf{G}[1]$.

\end{proof}

\noindent\textbf{Claim 2.} For any $k$, $x_{1}^{k}+x_{p+1}^{k}\geq
y_{1}^{k}+y_{p+1}^{k}$.

\begin{proof}
    
Note that if $x_{1}^{m}+x_{p+1}^{m}\geq y_{1}^{m}+y_{p+1}^{m}$, together
with Claim 1, we have 
\begin{align}
x_{t-\bar{p}+1}^{m+1}& =x_{1}^{m}+x_{p+1}^{m}+\beta x_{i}^{m}+\gamma
x_{j}^{m}  \notag \\
& \geq y_{1}^{m}+\alpha y_{p+1}^{m}+\beta y_{i}^{m}+\gamma y_{j}^{m}=y_{t-%
\bar{p}+1}^{m+1}.  \label{eq:7}
\end{align}%
Therefore, 
\begin{eqnarray*}
x_{1}^{m+1}+x_{p+1}^{m+1} &=&x_{1}^{m}+x_{p+1}^{m}+2x_{t-\bar{p}%
+1}^{m}+2x_{i}^{m}+(2\beta -2)x_{i}^{m}+\gamma x_{j}^{m} \\
&\geq &y_{1}^{m}+y_{p+1}^{m}+2y_{t-p+1}^{m}+2y_{i}^{m}+(2\beta
-2)y_{i}^{m}+\gamma y_{j}^{m} \\
&\geq &y_{1}^{m}+y_{p+1}^{m}+(1+\alpha )y_{t-p+1}^{m}+(1-\alpha
)y_{p+2}^{m}+2\beta y_{i}^{m}+\gamma y_{j}^{m} \\
&=&y_{1}^{m+1}+y_{p+1}^{m+1},
\end{eqnarray*}%
where the first inequality comes from the inductive assumption and the
second inequality follows from the fact that $y_{t-p+1}^{k}\geq y_{p+2}^{k}$
which is straightforward due to neighborhood nesting.

\end{proof}

\noindent\textbf{Claim 3.} For any $k$, $x_{t-p+1}^{k}+x_{p+1}^{k}\geq
y_{1}^{k}+y_{p+2}^{k}$.

\begin{proof}

Note that, since $\alpha \in \lbrack 0,1]$,  
\begin{equation*}
y_{1}^{m+1}+y_{p+2}^{m+1}\leq y_{1}^{m}+y_{p+2}^{m}+y_{p+1}^{m}+y_{t-\bar{p}%
+1}^{m}+\beta y_{i}^{m}+\gamma y_{j}^{m}\text{.}
\end{equation*}%
Moreover 
\begin{equation*}
x_{t-p+1}^{m+1}+x_{p+1}^{m+1}=2x_{1}^{m}+x_{p+1}^{m}+x_{t-\bar{p}%
+1}^{m}+2\beta x_{i}^{m}+\gamma x_{j}^{m}\text{.}
\end{equation*}%
By Claims 1 and 2, it is sufficient to show that 
\begin{equation*}
x_{1}^{m}+x_{t-\bar{p}+1}^{m}\geq y_{p+2}^{m}+y_{t-\bar{p}+1}^{m}.
\end{equation*}%
This holds since by Claim 1, $x_{1}^{m}+x_{t-\bar{p}+1}^{m}\geq
y_{1}^{m}+y_{t-\bar{p}+1}^{m}$ and $y_{1}^{m}\geq y_{p+2}^{m}$.

\end{proof}

\bigskip

Now, we are going to prove that $\sum_{k\in N}x_{k}^{m}\geq \sum_{k\in
N}y_{k}^{m}$ with the three claims. 
\begin{eqnarray*}
\sum_{k\in N}x_{k}^{m} &=&x_{1}^{m}+x_{p+1}^{m}+x_{t-\bar{p}+1}^{m}+\beta
x_{i}^{m}+\gamma x_{j}^{m} \\
\sum_{k\in N}y_{k}^{m} &=&y_{1}^{m}+y_{p+1}^{m}+y_{t-\bar{p}%
+1}^{m}+y_{p+2}^{m}+\beta y_{i}^{m}+\gamma y_{j}^{m}
\end{eqnarray*}
By Claim 1, we have $\sum_{k\in N}x_{k}^{m}\geq \sum_{k\in N}y_{k}^{m}$
whenever

\begin{equation}
x_{1}^{m}+x_{p+1}^{m}+x_{t-\bar{p}+1}^{m}+x_{i}^{m}\geq
y_{1}^{m}+y_{p+1}^{m}+y_{t-\bar{p}+1}^{m}+y_{p+2}^{m}+y_{i}^{m}  \label{eq:8}
\end{equation}

Decomposing the right-hand side 
\begin{eqnarray*}
RHS &=&(4-\alpha )y_{1}^{m-1}+(2+\alpha )y_{t-\bar{p}+1}^{m-1}+(2+\alpha
)y_{p+1}^{m-1}+(1-\alpha )y_{p+2}^{m-1}+(4\beta -1)y_{i}^{m-1}+3\gamma
y_{j}^{m-1} \\
&=&(1-\alpha )(y_{1}^{m-1}+y_{p+2}^{m-1})+(2+\alpha )(y_{t-\bar{p}%
+1}^{m-1}+y_{p+1}^{m-1})+3y_{1}^{m-1}+(4\beta -1)y_{i}^{m-1}+3\gamma
y_{j}^{m-1} \\
&\leq &(1-\alpha )(x_{t-p+1}^{m-1}+x_{p+1}^{m-1})+(2+\alpha )(x_{t-\bar{p}%
+1}^{m-1}+x_{p+1}^{m-1})+3y_{1}^{m-1}+(4\beta -1)y_{i}^{m-1}+3\gamma
y_{j}^{m-1} \\
&=&3(x_{t-p+1}^{m-1}+x_{p+1}^{m-1})+3y_{1}^{m-1}+(4\beta
-1)y_{i}^{m-1}+3\gamma y_{j}^{m-1}
\end{eqnarray*}%
where the inequality above comes from Claim 3. The left-hand side of
equation \eqref{eq:8} is 
\begin{equation*}
LHS=3x_{1}^{m-1}+3x_{t-\bar{p}+1}^{m-1}+3x_{p+1}^{m-1}+(4\beta
-1)x_{i}^{m-1}+3\gamma x_{j}^{m-1}.
\end{equation*}%
Therefore, by Claim 1 and inequality \eqref{eq:6}, $LHS\geq RHS$.

For strictness, note that the decomposition of RHS above implies that when $%
y_{1}^{m-1}+y_{p+2}^{m-1}<x_{t-p+1}^{m-1}+x_{p+1}^{m-1}$, $\sum_{k\in
N}x_{k}^{m}>\sum_{k\in N}y_{k}^{m}$. For the case of $m=2$, it is
straightforward to show that $%
y_{1}^{1}+y_{p+2}^{1}<x_{t-p+1}^{1}+x_{p+1}^{1} $ when $\alpha <1$. We
complete the proof for the strict inequality in the proposition by showing
that 
\begin{equation*}
x_{1}^{m}+x_{p+1}^{m}>y_{1}^{m}+y_{p+1}^{m}\Rightarrow
x_{1}^{m+1}+x_{p+1}^{m+1}>y_{1}^{m+1}+y_{p+1}^{m+1}\text{.}
\end{equation*}%
Then, since $x_{1}^{m+1}=x_{t-p+1}^{m+1}$ and $y_{p+1}^{m+1}>y_{p+2}^{m+1}$,
it implies $x_{t-p+1}^{m-1}+x_{p+1}^{m-1}>y_{1}^{m+1}+y_{p+2}^{m+1}$. The
proof is similar to that of Claim 3, except in the last step we strength the
inequality from $\geq $ to $>$ since $y_{1}^{m}>y_{p+2}^{m}$ when $\alpha <1$%
.

\hfill $\Box $

Now, we will show that $\lambda_{\max}(\mathbf{QC})=\lambda>\lambda_{\max}(%
\mathbf{\hat{G}})=\hat{\lambda}$. Let $\mathbf{x}$ and $\mathbf{y}$ be the
eigencentrality vectors of nodes in $\mathbf{QC}$ and $\mathbf{\hat{G}}$
respectively. Then we have 
\begin{equation}
\lambda \mathbf{x=QC\cdot x}\text{ and }\hat{\lambda}\mathbf{y=\hat{G}y}%
\text{.}  \label{eq:eigen}
\end{equation}

Let $\bar{p}=\frac{p(p-1)}{2}$ as in the previous proof. In the
quasi-complete graph, nodes are classified into three classes according to
automorphic equivalence: nodes $1$ to $t-\bar{p}+1$, nodes $t-\bar{p}+2$ to $%
p$, and node $p+1$. We use index $1$ for the first class of nodes and index $%
i$ for the second class. There are $\beta =t-\bar{p}+1\geq 1$ nodes in the
first class and $\gamma =p-\beta \geq 1$ nodes in the second class. By
equation \eqref{eq:eigen}, the following equalities hold for the
quasi-complete graph: 
\begin{eqnarray*}
\lambda x_{1} &=&(\beta -1)x_{1}+x_{p+1}+\gamma x_{j} \\
\lambda x_{j} &=&\beta x_{1}+(\gamma -1)x_{j} \\
\lambda x_{p+1} &=&\beta x_{1}
\end{eqnarray*}

Solving this linear equation system yields: 
\begin{equation}
\lambda=\beta-1+\frac{\gamma\beta}{\lambda-(\gamma-1)}+\frac{\beta}{\lambda}%
=f(\lambda)  \label{eq:lamuda}
\end{equation}

For the other graph $\mathbf{\hat{G}}$, there are five classes of nodes
according to automorphic equivalence: Node 1; nodes 2 to node $t-\bar{p}$;
nodes $t-\bar{p}+1$ to $p$; and nodes $p+1$ and $p+2$. We use index $i$ to
represent the second class, which contains $\beta -2$ nodes. Index $j$
represents the third class, containing $\gamma +1$ nodes. By equation %
\eqref{eq:eigen}, the following equalities hold for graph $\mathbf{\hat{G}}$%
: 
\begin{eqnarray}
\hat{\lambda}y_{1} &=&(\beta -2)y_{i}+(\gamma +1)y_{j}+y_{p+1}+y_{p+2}
\label{eq:spe1} \\
\hat{\lambda}y_{i} &=&y_{1}+(\beta -3)y_{i}+(\gamma +1)y_{j}+y_{p+1}
\label{eq:spe2} \\
\hat{\lambda}y_{j} &=&y_{1}+(\beta -2)y_{i}+\gamma y_{j}  \label{eq:spe3} \\
\hat{\lambda}y_{p+1} &=&y_{1}+(\beta -2)y_{i}  \label{eq:spe4} \\
\hat{\lambda}y_{p+2} &=&y_{1}  \label{eq:spe5}
\end{eqnarray}

Subtracting equation \eqref{eq:spe2} from equation \eqref{eq:spe1} leads to: 
$\left( \hat{\lambda}+1-\frac{1}{\hat{\lambda}}\right) y_{1}=\left( \hat{%
\lambda}+1\right) y_{i}$

Combining equations \eqref{eq:spe4} and \eqref{eq:spe5}, we get $y_{i}=\frac{%
\hat{\lambda}y_{p+1}-y_{1}}{(\beta-2)}$. Substituting this into the equation
above yields: 
\begin{equation*}
y_{p+1}=\frac{(\beta-2)(\hat{\lambda}+1-\frac{1}{\hat{\lambda}})y_{1}}{\hat{%
\lambda}(\hat{\lambda}+1)}+\frac{y_{1}}{\hat{\lambda}}
\end{equation*}

From equations \eqref{eq:spe3} and \eqref{eq:spe4}, we derive $(\hat{\lambda}%
-\gamma)y_{j}=\hat{\lambda}y_{p+1}$. Combining with the equality above: 
\begin{equation*}
y_{j}=\frac{\hat{\lambda}}{\hat{\lambda}-\gamma}y_{p+1}=\frac{(\beta-2)(\hat{%
\lambda}+1-\frac{1}{\hat{\lambda}})y_{1}}{(\hat{\lambda}-\gamma)(\hat{\lambda%
}+1)}+\frac{y_{1}}{\hat{\lambda}-\gamma}
\end{equation*}

Substituting these expressions for $y_{p+1}$ and $y_{j}$ into equation %
\eqref{eq:spe1} yields: 
\begin{equation}
\hat{\lambda}=\frac{(\beta-1)(\gamma+1)(\hat{\lambda}+1-\frac{1}{\hat{\lambda%
}})+(\gamma+1)(1+\frac{1}{\hat{\lambda}})}{(\hat{\lambda}-\gamma)(\hat{%
\lambda}+1)}+\frac{2}{\hat{\lambda}}=g(\hat{\lambda})  \label{eq:lamudah}
\end{equation}

Since $\mathbf{1}^{\prime }\mathbf{QC}^{k}\mathbf{1}>\mathbf{1}^{\prime }%
\mathbf{\hat{G}}^{k}\mathbf{1}$ for any $k\geq 2$, we have $\lambda \geq 
\hat{\lambda}$. Therefore, we conclude that $\lambda >\hat{\lambda}$
whenever $f(x)-g(x)>0$ for any $x>\hat{\lambda}$, where the functions $f(x)$
and $g(x)$ are defined by equations \eqref{eq:lamuda} and \eqref{eq:lamudah}%
, respectively.

It can be shown that: 
\begin{equation*}
f(x)-g(x)=(\beta -2)A+\frac{(\gamma +x+1)(x+1)}{(x-\gamma +1)(x-\gamma )}
\end{equation*}%
where 
\begin{align*}
A& =\frac{x+1}{x-\gamma +1}-\frac{(\gamma +1)(x+1-\frac{1}{x})}{(x-\gamma
)(x+1)}+\frac{1}{x} \\
& >\frac{x+1}{x-\gamma +1}-\frac{\gamma +1}{x-\gamma }+\frac{1}{x} \\
& =\frac{x^{2}+\gamma ^{2}-2\gamma x-\gamma -1}{(x-\gamma +1)(x-\gamma )}+%
\frac{1}{x} \\
& \geq \frac{1}{x}-\frac{\gamma +1}{(x-\gamma +1)(x-\gamma )}
\end{align*}

As a result: 
\begin{equation*}
f(x)-g(x)>\frac{\beta -2}{x}+\frac{(\gamma +x+1)(x+1)}{(x-\gamma
+1)(x-\gamma )}-\frac{(\beta -2)(\gamma+1) }{(x-\gamma +1)(x-\gamma )}\text{.%
}
\end{equation*}%
When $\beta \geq 2$, the right-hand side is non-negative since $x>\gamma
+\beta -1$, where $\gamma +\beta -1$ is the spectral radius of the clique
with size $\gamma +\beta $. When $\beta =1$, the right-hand side is given by%
\begin{equation*}
\frac{x^{3}+x^{2}+\gamma x^{2}+4\gamma x+x-\gamma ^{2}+\gamma }{(x-\gamma
+1)(x-\gamma )}\geq 0\text{.}
\end{equation*}
\hfill $\Box $

\bigskip

\noindent \textbf{Proof of Theorem \ref{thm:unweighted-QC}.} 

The first part of the theorem, that the myopic optimum coincides with the greedy algorithm, follows directly from the definition of myopic optimality.

The second part of Theorem \ref{thm:unweighted-QC}, that the greedy algorithm forms a QC graph at each step, follows directly from Lemma \ref{lem:unweighted-QC} which discriminates the two NSGs. \hfill$\Box$

\bigskip

\noindent \textbf{Proof of Corollary \ref{cor:comparative}.} Recall that for 
$0<\phi <1/\lambda _{\max }(\mathbf{G})$ and any integer $\alpha \geq 0$, 
\begin{equation*}
b(\alpha ,\phi ,\mathbf{G})=\sum_{k=0}^{\infty }\binom{\alpha +k-1}{k}\phi ^{k}\,W^{k}(\mathbf{G}),
\end{equation*}%
and for $\beta >0$, 
\begin{equation*}
c(\beta ,\mathbf{G})=\sum_{k=0}^{\infty }\frac{\beta ^{k}}{k!}\,W^{k}(%
\mathbf{G}).
\end{equation*}%
We will also use the complete graph $\mathbf{C}$ on $n$ nodes, for which $%
W^{k}(\mathbf{C})=\mathbf{1}^{\prime }\mathbf{C}^{k}\mathbf{1}=n(n-1)^{k}$.

We first analyze the case $\delta\to+\infty$ (a far-sighted planner). In
that case, the dynamic problem at a given total number of links $T$ reduces
to a static problem: 
\begin{equation*}
\max_{\mathbf{G}\in\mathcal{G}(T)} u(\mathbf{G}),\quad \text{where }
u\in\{b(\alpha,\phi,\cdot),\,c(\beta,\cdot)\}.
\end{equation*}

\smallskip By Theorem 2 and Corollary 2 of \cite{Abrego2009}, for $n\ge 6$:
- when $T\in[4,\frac{n^2-3n}{4})$, the quasi-star network $\mathbf{QS}(T)$
uniquely maximizes $\sum_i d_i^2=W^2(\mathbf{G})$ among $\mathcal{G}(T)$; -
when $T\in(\frac{n^2+n}{4},\frac{n^2-n}{2}]$, the quasi-complete network $%
\mathbf{QC}(T)$ uniquely maximizes $W^2(\mathbf{G})$ among $\mathcal{G}(T)$.
Hence, there exist gaps $\varepsilon_1,\varepsilon_2>0$ such that for all $%
\mathbf{G}\in\mathcal{G}(T)$: 
\begin{equation*}
\begin{cases}
W^2(\mathbf{QS}(T)) \ge W^2(\mathbf{G})+\varepsilon_1, & \text{if } T\in[4,%
\frac{n^2-3n}{4}), \\[2pt] 
W^2(\mathbf{QC}(T)) \ge W^2(\mathbf{G})+\varepsilon_2, & \text{if } T\in(%
\frac{n^2+n}{4},\frac{n^2-n}{2}].%
\end{cases}%
\end{equation*}

\smallskip Let $(\rho_k)_{k\ge 0}$ be any nonnegative weights with $\rho_k>0$
for all $k\ge 0$. For any $\mathbf{G}\in\mathcal{G}(T)$, 
\begin{equation*}
\sum_{k=0}^{\infty}\rho_k W^k(\mathbf{G}) = \sum_{k=0}^{2}\rho_k W^k(\mathbf{%
G}) + \sum_{k=3}^{\infty}\rho_k W^k(\mathbf{G}) \le \sum_{k=0}^{2}\rho_k W^k(%
\mathbf{G}) + \sum_{k=3}^{\infty}\rho_k W^k(\mathbf{C}),
\end{equation*}
since $W^k(\mathbf{G})\le W^k(\mathbf{C})$ for all $k$ (the complete graph
maximizes walk counts of every length at fixed $n$). Define 
\begin{equation*}
S(\rho):=\sum_{k=3}^{\infty}\rho_k W^k(\mathbf{C})=
n\sum_{k=3}^{\infty}\rho_k (n-1)^k.
\end{equation*}
Then for all $\mathbf{G}\in\mathcal{G}(T)$, 
\begin{equation}  \label{eq:master-bound}
\sum_{k=0}^{\infty}\rho_k W^k(\mathbf{G}) \le \sum_{k=0}^{2}\rho_k W^k(%
\mathbf{G}) + S(\rho).
\end{equation}

\smallskip \emph{Case 1: $u(\mathbf{G})=b(\alpha,\phi,\mathbf{G})$.} Here $%
\rho_k=\binom{\alpha+k-1}{k}\phi^k$ and, for $|\phi(n-1)|<1$, 
\begin{equation*}
\sum_{k=0}^{\infty}\rho_k (n-1)^k = \sum_{k=0}^{\infty}\binom{\alpha+k-1}{k} %
\bigl(\phi(n-1)\bigr)^k = (1-\phi(n-1))^{-\alpha}.
\end{equation*}
Therefore, 
\begin{equation*}
S(\rho)= n\sum_{k=3}^{\infty}\rho_k (n-1)^k = n\Bigl[(1-\phi(n-1))^{-\alpha}
- \sum_{k=0}^{2}\binom{\alpha+k-1}{k}\bigl(\phi(n-1)\bigr)^k\Bigr] < n(\frac{%
1}{(1-\phi(n-1))^{\alpha}}-1).
\end{equation*}
Using \eqref{eq:master-bound}, we get for all $\mathbf{G}$: 
\begin{equation*}
b(\alpha,\phi,\mathbf{G}) < \sum_{k=0}^{2}\rho_k W^k(\mathbf{G}) + n(\frac{1%
}{(1-\phi(n-1))^{\alpha}}-1).
\end{equation*}
Fix $T\in[4,\frac{n^2-3n}{4})$. Let $\varepsilon_1>0$ be the gap for $W^2$
identified above. If $\phi>0$ is small enough so that 
\begin{equation*}
n(\frac{1}{(1-\phi(n-1))^{\alpha}}-1) \le \varepsilon_1
\quad\Longleftrightarrow\quad (1-\phi(n-1))^{-\alpha}\le \frac{\varepsilon_1%
}{n}-1 \quad\Longleftrightarrow\quad \phi \le \frac{1-(%
\varepsilon_1/n-1)^{-1/\alpha}}{n-1},
\end{equation*}
then for any $\mathbf{G}\in\mathcal{G}(T)$, 
\begin{equation*}
b(\alpha,\phi,\mathbf{G}) < \sum_{k=0}^{2}\rho_k W^k(\mathbf{G}) +
\varepsilon_1 \le \sum_{k=0}^{2}\rho_k W^k\bigl(\mathbf{QS}(T)\bigr) \le b%
\bigl(\alpha,\phi,\mathbf{QS}(T)\bigr).
\end{equation*}

The proof of the remaining part of 1(a) is analogous: for $T\in (\frac{%
n^{2}+n}{4},\frac{n^{2}-n}{2}]$, choosing $\phi >0$ small enough so that $%
n(1-\phi (n-1))^{-\alpha }\leq \varepsilon _{2}$ yields $\mathbf{QC}(T)$ as
the maximizer.

\smallskip \emph{Case 2: $u(\mathbf{G})=c(\beta ,\mathbf{G})$.} Here $\rho
_{k}=\beta ^{k}/k!$. Thus 
\begin{equation*}
S(\rho )=n\sum_{k=3}^{\infty }\frac{\beta ^{k}}{k!}(n-1)^{k}=n\Bigl(e^{\beta
(n-1)}-1-\beta (n-1)-\tfrac{1}{2}\beta ^{2}(n-1)^{2}\Bigr)<ne^{\beta
(n-1)}-n.
\end{equation*}%
Therefore, by \eqref{eq:master-bound}, 
\begin{equation*}
c(\beta ,\mathbf{G})<\sum_{k=0}^{2}\frac{\beta ^{k}}{k!}W^{k}(\mathbf{G}%
)+ne^{\beta (n-1)}-n.
\end{equation*}%
Fix $T\in \lbrack 4,\frac{n^{2}-3n}{4})$, and let $\varepsilon _{1}>0$ be as
above. If $\beta >0$ satisfies 
\begin{equation*}
ne^{\beta (n-1)}-n\leq \varepsilon _{1}\quad \Longleftrightarrow \quad \beta
\leq \frac{1}{n-1}\ln \!\Bigl(\frac{\varepsilon _{1}+n}{n}\Bigr),
\end{equation*}%
then for any $\mathbf{G}\in \mathcal{G}(T)$, 
\begin{equation*}
c(\beta ,\mathbf{G})<\sum_{k=0}^{2}\frac{\beta ^{k}}{k!}W^{k}(\mathbf{G}%
)+\varepsilon _{1}\leq \sum_{k=0}^{2}\frac{\beta ^{k}}{k!}W^{k}\bigl(\mathbf{%
QS}(T)\bigr)\leq c\bigl(\beta ,\mathbf{QS}(T)\bigr).
\end{equation*}%
Thus $\mathbf{QS}(T)$ is optimal. The proof of the remaining part is
analogous: for $T\in (\frac{n^{2}+n}{4},\frac{n^{2}-n}{2}]$, if $\beta \leq 
\frac{1}{n-1}\ln (\varepsilon _{2}/n+1)$, then $\mathbf{QC}(T)$ is optimal.

\smallskip \emph{Myopic case $\delta \rightarrow 0^{+}$.} Statements 1(b)
and 2(c) follow from Theorem \ref{thm:unweighted-QC}, which asserts that the
myopic maximizer is a quasi-complete network, irrespectively of $\phi $ (for 
$b$) or $\beta $ (for $c$).

\smallskip \emph{Spectral dominance for large $\beta $ (statement 2(b)).} We
claim that if $\lambda _{\max }(\mathbf{G})<\lambda _{\max }(\hat{\mathbf{G}}%
)$, then there exists $\bar{\beta}<\infty $ such that for all $\beta \geq 
\bar{\beta}$, $c(\beta ,\hat{\mathbf{G}})>c(\beta ,\mathbf{G})$. By
Perron--Frobenius asymptotics, 
\begin{equation*}
\lim_{k\rightarrow \infty }\frac{W^{k}(\hat{\mathbf{G}})}{W^{k}(\mathbf{G})}%
=\infty \quad \text{iff}\quad \lambda _{\max }(\hat{\mathbf{G}})>\lambda
_{\max }(\mathbf{G}).
\end{equation*}%
Thus for any $M>1$, there exist $\bar{k}\in \mathbb{N}$ such that for all $%
k\geq \bar{k}$, $W^{k}(\hat{\mathbf{G}})\geq M\,W^{k}(\mathbf{G})$. Then 
\begin{align*}
c(\beta ,\hat{\mathbf{G}})-c(\beta ,\mathbf{G})& =\sum_{k=0}^{\infty }\frac{%
\beta ^{k}}{k!}\bigl(W^{k}(\hat{\mathbf{G}})-W^{k}(\mathbf{G})\bigr) \\
& \geq \sum_{k=0}^{\bar{k}-1}\frac{\beta ^{k}}{k!}\bigl(W^{k}(\hat{\mathbf{G}%
})-W^{k}(\mathbf{G})\bigr)\;+\;(M-1)\sum_{k=\bar{k}}^{\infty }\frac{\beta
^{k}}{k!}W^{k}(\mathbf{G}).
\end{align*}%
Let 
\begin{equation*}
A:=\sum_{k=0}^{\bar{k}-1}\frac{1}{k!}\bigl|W^{k}(\hat{\mathbf{G}})-W^{k}(%
\mathbf{G})\bigr|\quad \text{and}\quad B:=\frac{W^{\bar{k}}(\mathbf{G})}{%
\bar{k}!}>0.
\end{equation*}%
Then for all $\beta \geq 1$, 
\begin{equation*}
\left\vert \sum_{k=0}^{\bar{k}-1}\frac{\beta ^{k}}{k!}\bigl(W^{k}(\hat{%
\mathbf{G}})-W^{k}(\mathbf{G})\bigr)\right\vert \leq \beta ^{\bar{k}%
-1}A,\qquad \sum_{k=\bar{k}}^{\infty }\frac{\beta ^{k}}{k!}W^{k}(\mathbf{G}%
)\geq \frac{\beta ^{\bar{k}}}{\bar{k}!}W^{\bar{k}}(\mathbf{G})=\beta ^{\bar{k%
}}B.
\end{equation*}%
Hence, if 
\begin{equation*}
\beta \;\geq \;\bar{\beta}:=\frac{A}{(M-1)B},
\end{equation*}%
we obtain $c(\beta ,\hat{\mathbf{G}})-c(\beta ,\mathbf{G})>0$. This proves
2(b).

\smallskip Combining all parts yields the corollary. \qed
\bigskip

\noindent \textbf{Proof of Proposition \ref{pro:weighted-KB2-KB}.} The gist
of the proof is to show that the solution of problem (\ref{design problem
weighted 1}) an extreme point of the set of feasible network formation paths 
$S_{w}$, which coincides with the set of sequences of unweighted networks.

\noindent

\begin{claim}
\label{claim:convex} The set $S_{w}$ is a convex set, and the extreme points
of $S_{w}$ are unweighted networks.
\end{claim}

To show the convexity of $S_{w}$, consider two elements $\mathbf{s}%
_{w}=\left( \mathbf{G}(t)\right)_{t=1}^{T},\mathbf{\hat{s}}_{w}=\left(
\mathbf{\hat{G}}(t)\right)_{t=1}^{T}\in S_{w}$ and a constant $\alpha \in
(0,1)$. It is easy to verify that for any $t\geq 1$,
\begin{equation*}
\mathbf{1}^{\prime}\left[\alpha\mathbf{G}(t)+(1-\alpha)\mathbf{\hat{G}}%
(t)-\alpha\mathbf{G}(t-1)-(1-\alpha)\mathbf{\hat{G}}(t-1)\right]\mathbf{1}%
=2.
\end{equation*}%
That is, $\alpha\mathbf{G}(t)+(1-\alpha)\mathbf{\hat{G}}(t)$ can be
obtained by adding one unit of weight from $\alpha\mathbf{G}(t-1)+(1-\alpha
)\mathbf{\hat{G}}(t-1)$. As a result, $\alpha\mathbf{s}_{w}+(1-\alpha)%
\mathbf{\hat{s}}_{w}\in S_{w}$, and thus $S_{w}$ is convex.

Next, we show that the extreme points of $S_{w}$ satisfy ext$(S_{w})=S$%
, where $S$ is the set of feasible unweighted network formation paths. Since
$S\subseteq$ext$\left(S_{w}\right)$, we only need to prove ext$\left(
S_{w}\right)\subseteq S$ by showing that $(S_{w}\backslash S)\cap\text{ext}%
(S_{w})=\emptyset$. Let $\mathbf{s}_{w}\in S_{w}\backslash S$. Then there
exists some $t\leq T$ such that $\mathbf{G}(t)$ is a strictly weighted
network. Let $t^{\prime}$ be the earliest such period. Hence, there exist
some $i,j,k,l$ such that $g_{ij}(t^{\prime})=g_{ji}(t^{\prime})\in(0,1)$,
$g_{kl}(t^{\prime})=g_{lk}(t^{\prime})\in(0,1)$, and $(i,j)\neq(k,l)$. We
construct two sequences of weighted networks $\mathbf{s}^{+}=\left(
\mathbf{G}^{+}\left(t\right)\right)_{t=1}^{T}$ and $\mathbf{s}^{-}=\left(
\mathbf{G}^{-}\left(t\right)\right)_{t=1}^{T}$ in $S_{w}$ such that $%
\mathbf{s}_{w}=\frac{1}{2}\left(\mathbf{s}^{+}+\mathbf{s}^{-}\right)$.

\medskip
\noindent\textbf{Case 1.} $g_{ij}(t)=g_{ji}(t)\in(0,1)$ and $g_{kl}(t)=g_{lk}(t)\in
(0,1)$ for all $t^{\prime}\leq t\leq T$. 

Let $\mathbf{G}^{+}(t)=%
\mathbf{G}^{-}(t)=\mathbf{G}(t)$ for any $t<t^{\prime}$, and for any $t\geq
t^{\prime}$ let
\begin{eqnarray}
g_{ij}^{+}(t) &=&g_{ji}^{+}(t)=g_{ij}(t)-\Delta;\quad
g_{kl}^{+}(t)=g_{lk}^{+}(t)=g_{kl}(t)+\Delta \notag \\
g_{ij}^{-}(t) &=&g_{ji}^{-}(t)=g_{ij}(t)+\Delta;\quad
g_{kl}^{-}(t)=g_{lk}^{-}(t)=g_{kl}(t)-\Delta \notag \\
g_{pq}^{+}(t) &=&g_{pq}^{-}(t)=g_{pq}(t),\quad\forall(p,q)\notin\{(i,j),(k,l)\},
\label{eq:construction}
\end{eqnarray}
where $\Delta=\min\{g_{ij}(t^{\prime}), g_{kl}(t^{\prime}), 1-g_{ij}(T), 1-g_{kl}(T)\}$.

By this construction, both $\mathbf{s}^{+}$ and $\mathbf{s}%
^{-}$ belong to $S_{w}$, and $\frac{1}{2}\left(\mathbf{G}^{+}\left(t\right)
+\mathbf{G}^{-}\left(t\right)\right)=\mathbf{G}(t)$ for any $t$. Thus, $%
\mathbf{s}_{w}\notin$ext$(S_{w})$.

\medskip
\noindent\textbf{Case 2.} Links $\left(i,j\right)$ and $\left(k,l\right)$ reach weight $1$
at the same time, i.e., $g_{ij}(\bar{t})=g_{kl}(\bar{t})=1$ and $%
g_{ij}(t)\in\left(0,1\right), g_{kl}(t)\in\left(0,1\right)$ for all $%
t^{\prime}\leq t<\bar{t}$. 

Let $\mathbf{G}^{+}(t)=\mathbf{G}^{-}(t)=%
\mathbf{G}(t)$ for any $t<t^{\prime}$ and $t\geq\bar{t}$. For $t^{\prime
}\leq t<\bar{t}$, construct $\mathbf{G}^{+}(t)$ and $\mathbf{G}^{-}(t)$ as in (\ref{eq:construction}). Then it is straightforward to show that both $\mathbf{s}^{+}$
and $\mathbf{s}^{-}$ belong to $S_{w}$, and $\frac{1}{2}\left(\mathbf{G}%
^{+}\left(t\right)+\mathbf{G}^{-}\left(t\right)\right)=\mathbf{G}(t)$
for any $t$.

\medskip
\noindent\textbf{Case 3.} One or both of links $\left(i,j\right)$ and $\left(k,l\right)$ reach weight $1$ at the end $\mathbf{G}(T)$, but at different time.

Without loss of generality, assume $\left(i,j\right)$ reaches $1$ at time $\bar{t}$ before $\left(k,l\right)$. That
is, $g_{ij}(\bar{t})=g_{ji}(\bar{t})=1$, $g_{ij}(t)=g_{ji}(t)\in(0,1)$ for
all $t^{\prime}\leq t<\bar{t}$, and $g_{kl}(\bar{t})=g_{lk}(\bar{t})\in(0,1)
$. Let $\mathbf{G}^{+}(t)=\mathbf{G}^{-}(t)=\mathbf{G}(t)$ for any $%
t<t^{\prime}$, and for any $t^{\prime}\leq t<\bar{t}$, construct $\mathbf{G}%
^{+}(t)$ and $\mathbf{G}^{-}(t)$ as in (\ref{eq:construction}). At time $\bar{t}$, since $%
g_{kl}(\bar{t})\in\left(0,1\right)$, there must exist another pair of
nodes $\left(i^{\prime},j^{\prime}\right)$ such that $g_{i^{\prime
}j^{\prime}}(\bar{t})\in\left(0,1\right)$. We then repeat the construction for $t>\bar{t}$ by analyzing the pairs $\left(i^{\prime
}, j^{\prime}\right)$ and $\left(k,l\right)$.

Consequently, we have $\mathbf{s}_{w}\notin\text{ext}(S_{w})$.

We complete the proof of Proposition \ref{pro:weighted-KB2-KB} using Claim %
\ref{claim:convex}. Given two sequences of weighted networks $\mathbf{s}%
_{w},\mathbf{\hat{s}}_{w}\in S_{w}$ and a constant $\alpha\in(0,1)$, the
following holds:
\begin{align*}
v\left(\alpha\mathbf{s}_{w}+\left(1-\alpha\right)\mathbf{\hat{s}}%
_{w}\right) &=\sum_{t=1}^{T}D(t)\cdot b\left(\phi,\alpha
\mathbf{G}(t)+\left(1-\alpha\right)\mathbf{\hat{G}}(t)\right) \\
&\leq\sum_{t=1}^{T}D(t)\cdot\left[\alpha b\left(\phi,\mathbf{G}(t)\right)+(1-\alpha)b\left(\phi,\mathbf{\hat{G}}(t)\right)\right] \\
&=\alpha v(\mathbf{s}_{w})+(1-\alpha)v(\mathbf{\hat{s}}_{w}).
\end{align*}%
The first inequality follows from Lemma \ref{lem:convexity}, the proof of
which is given by \cite{Sun2023}'s Lemma A.2. Combining this inequality with
Claim \ref{claim:convex}, we conclude that, even if we allow for weighted
networks, it is without loss of optimality to restrict to the set of
unweighted network sequences $S$ in solving the optimization problem. \hfill$\Box$

\bigskip

\noindent\textbf{Proof of Lemma \ref{lem:sun2023}.} We split the argument
into two cases.

\smallskip \emph{Case 1: $b_i(1,\mathbf{G}) > b_j(1,\mathbf{G})$.} This is
exactly Proposition 4 in \cite{Sun2023}: transferring weight from $j$ to $i$
so that in the post-reallocation network $\widehat{\mathbf{G}}$ one has $%
\hat g_{ik} \ge \hat g_{jk}$ for all $k\notin\{i,j\}$ strictly increases the
aggregate square of Katz-Bonacich centrality, i.e., $b(2,%
\widehat{\mathbf{G}}) > b(2,\mathbf{G})$.

\smallskip \emph{Case 2: $b_i(1,\mathbf{G}) < b_j(1,\mathbf{G})$.} Let $%
\widehat{\mathbf{G}}$ be any post-reallocation network obtained by moving
weight from $j$ to $i$ such that $\hat g_{ik}\ge g_{jk}$ for all $%
k\notin\{i,j\}$. We will reduce this case to Case 1 by symmetry.

Construct an auxiliary network $\overline{\mathbf{G}}$ from $\mathbf{G}$ by
transferring, for each $k\notin\{i,j\}$, the amount $\delta_k := \hat g_{ik}
- g_{jk} \ge 0$ \emph{from $i$ to $j$}, i.e., 
\begin{equation*}
\bar g_{jk} \;=\; g_{jk} + \delta_k \;=\; \hat g_{ik}, \qquad \bar g_{ik}
\;=\; g_{ik} - \delta_k,
\end{equation*}
and keep all other entries unchanged (respecting symmetry of the adjacency).
By construction, for every $k\notin\{i,j\}$, 
\begin{equation*}
\bar g_{jk} \;=\; \hat g_{ik} \;\ge\; g_{jk} \;\ge\; \hat g_{jk} \;=\; \bar
g_{ik}.
\end{equation*}
Thus, in $\overline{\mathbf{G}}$, node $j$ (weakly) dominates node $i$ in
the sense of having at least as large ties to all other nodes; moreover, we
are reallocating weight from $i$ to $j$.

Since $b_i(1,\mathbf{G}) < b_j(1,\mathbf{G})$, by applying the conclusion of
Case 1 with the roles of $(i,j)$ swapped, this reallocation (from the
lower-centrality node $i$ to the higher-centrality node $j$) strictly
increases the aggregate square of KB centralities: 
\begin{equation*}
b(2,\overline{\mathbf{G}}) \;>\; b(2,\mathbf{G}).
\end{equation*}

Finally, observe that $\widehat{\mathbf{G}}$ is isomorphic to $\overline{%
\mathbf{G}}$ via the relabeling that swaps $i$ and $j$. Because $b(2,\cdot)$
is invariant under node relabelings, we have 
\begin{equation*}
b(2,\widehat{\mathbf{G}}) \;=\; b(2,\overline{\mathbf{G}}) \;>\; b(2,\mathbf{%
G}),
\end{equation*}
which establishes the desired inequality in Case 2.

Combining the two cases completes the proof. \hfill$\Box$

\bigskip \noindent \textbf{Proof of Proposition \ref{pro:weighted-KB2} (i).}
Suppose the sequence of weighted networks $\mathbf{s}_{w}=\left( \mathbf{\ G}%
\left( t\right) \right) _{t=1}^{T}$ generates non-NSGs at some periods.
Denote $\left( \mathbf{W}\left( t\right) \right) _{t=1}^{T}$ the
weight-adding matrix, i.e., $\mathbf{G}\left( t\right) =\mathbf{G}\left(
t-1\right) +\mathbf{W}\left( t\right) $ for any $t$. Let $t^{\prime }$ be
the first time that $\mathbf{G}\left( t^{\prime }\right) $ is not a weighted
NSG. Consider two agents $i,j$ such that $i$ weight dominates $j$ in $%
\mathbf{G}\left( t\right) $ for any $t<t^{\prime }$ while $i$ does not
weight dominate $j$ in $\mathbf{G}\left( t^{\prime }\right) $. We construct
another sequence of networks $\hat{\mathbf{s}}_{w}=\left( \mathbf{\hat{G}}%
\left( t\right) \right) _{t=1}^{T}$, where $\mathbf{\hat{G}}\left(
t+1\right) =\mathbf{\hat{G}}\left( t\right) +\mathbf{\hat{W}}\left( t\right) 
$, according to the following rule,

\begin{enumerate}
\item For any $l\notin \left\{ i,j\right\} $, $\hat{w}_{il}\left( t\right)
=\min \left\{ 1-\hat{g}_{il}\left( t-1\right) \text{, }w_{il}\left( t\right)
+w_{jl}\left( t\right) \right\} $;\label{case_weighted_1}

\item For any $l\notin \left\{ i,j\right\} $, $\hat{w}_{jl}\left( t\right)
=\max \left\{ w_{il}\left( t\right) +w_{jl}\left( t\right) +\hat{g}%
_{il}\left( t-1\right) -1\text{, }0\right\} $;\label{case_weighted_2}

\item For any $k,l\notin \left\{ i,j\right\} $ or $\left( k,l\right) =\left(
i,j\right) $, $\hat{w}_{kl}\left( t\right) =w_{kl}\left( t\right) $.\label%
{case_weighted_3}
\end{enumerate}

According to the constructing rule, the weight assigned to $\left(j,k\right)$
is reallocated to $\left( i,k\right)$ preferentially. We first show that,
for any $t\geq t^\prime$, 
\begin{equation*}
\mathbf{\tilde{W}}\left( t\right) =\mathbf{\hat{G}}\left( t\right) - \mathbf{%
G}\left( t\right) =\overset{t}{\underset{s=t^\prime}{\sum }}\left[ \mathbf{%
\hat{W}}\left( s\right) -\mathbf{W}\left( s\right) \right]
\end{equation*}
is a weight reallocation from $j$ to $i$.

According to the
constructing rule, we have $\tilde{w}_{kl}\left( t\right) =0$ for any $k,l\notin \left\{ i,j\right\} $ or $\left( k,l\right)=\left( i,l\right) $. For any $k\notin \left\{ i,j\right\} $, we have%
\begin{eqnarray*}
\tilde{w}_{ik}\left( t\right) +\tilde{w}_{jk}\left( t\right) &=&\overset{t}{%
\underset{s=t^\prime}{\sum }}\left[ \left( \hat{w}_{ik}\left( s\right)
-w_{ik}\left( s\right) \right) +\left( \hat{w}_{jk}\left( s\right)
-w_{jk}\left( s\right) \right) \right] \\
&=&\overset{t}{\underset{s=t^\prime}{\sum }}\left[ \left( \hat{w}_{ik}\left(
s\right) +\hat{w}_{jk}\left( s\right) \right) -\left( w_{ik}\left( s\right)
+w_{jk}\left( s\right) \right) \right]=0
\end{eqnarray*}

Then, we argue that $\tilde{w}_{ik}\left( t\right) \geq 0$ for any $k\notin
\left\{ i,j\right\} $ and $t\geq t^\prime$. Suppose not. There exists $%
k\notin \left\{ i,j\right\} $ such that $\tilde{w}_{ik}\left( t\right) <0$.
That is, $\overset{t}{\underset{s=t^\prime}{\sum }}\hat{w}_{ik}\left(
s\right) <\overset{t}{\underset{s=t^\prime}{\sum }}w_{ik}\left( s\right) $.
Therefore, we further have%
\begin{equation*}
\hat{g}_{ik}\left( t\right) =g_{ik}\left( t^\prime-1\right) +\overset{}{%
\overset{t}{\underset{s=t^\prime}{\sum }}\tilde{w}_{ik}\left( s\right) }%
<g_{ik}\left( t^\prime-1\right) +\overset{}{\overset{t}{\underset{s=t^\prime}%
{\sum }}w_{ik}\left( s\right) =g_{ik}\left( t\right) \leq 1\text{.}}
\end{equation*}

Hence, for any $s\leq t$, $\hat{g}_{ik}\left( s\right) <1$. By the
construction of $\hat{g}_{ik}\left( s\right) $, we have $w_{ik}\left(
s\right) +w_{jk}\left( s\right) <1-g_{ik}\left( t^\prime-1\right) $. As a
result, $\tilde{w}_{ik}\left( t\right) =\overset{t}{\underset{s=t^\prime}{%
\sum }}\left( \hat{w}_{ik}\left( s\right) -w_{ik}\left( s\right) \right) =%
\overset{}{\overset{t}{\underset{s=t^\prime}{\sum }}w_{jk}\left( s\right)
\geq 0\text{.}}$ It contradicts the assumption that $\tilde{w}_{ik}\left(
t\right) <0$. We conclude that $\mathbf{\tilde{W}}\left( t\right) $ is a
weight reallocation from $j$ to $i$.

To apply Lemma \ref{lem:sun2023}, we will show that $\hat{g}_{ik}\left(
t\right) \geq g_{jk}\left( t\right) $ in the following. If $\hat{g}%
_{ik}\left( t\right) =1$, the inequality trivially holds. If $\hat{g}%
_{ik}\left( t\right) <1$, by the construction rule, we have $\hat{w}%
_{ik}\left( s\right) =w_{ik}\left( s\right) +w_{jk}\left( s\right) $ and $%
\hat{w}_{jk}\left( s\right) =0$ for any $s\in \left[ t^\prime,t\right] $.
Therefore,%
\begin{eqnarray*}
\hat{g}_{ik}\left( t\right) &=&g_{ik}\left( t^\prime\right) +\overset{t}{%
\underset{s=t^\prime}{\sum }}\hat{w}_{jk}\left( s\right) =g_{ik}\left(
t^\prime\right) +\overset{t}{\underset{s=t^\prime}{\sum }}\left(
w_{ik}\left( s\right) +w_{jk}\left( s\right) \right) \\
&\geq &g_{ik}\left( t^\prime\right) \geq g_{jk}\left( t^\prime\right)
=g_{jk}\left( t^\prime\right) +\overset{t}{\underset{s=t^\prime}{\sum }}\hat{%
w}_{jk}\left( s\right) =\hat{g}_{jk}\left( t\right) \text{.}
\end{eqnarray*}
To sum up, $\mathbf{\tilde{W}}\left( t\right) $ is a weight reallocation
from $j$ to $i$ that satisfies the conditions in Lemma \ref{lem:sun2023}. As
a result, for each period $t$, $\mathbf{\hat{G}}\left( t\right) $ generates
a higher payoff than $\mathbf{G}\left( t\right) $. That is, for any sequence
of networks $\mathbf{s}_{w}$ generating non-NSG in some periods, we can
construct a sequence of networks $\mathbf{\hat{s}}_{w}$ inducing higher $%
b(2,\cdot)$.\hfill$\Box$

\bigskip

\noindent \textbf{Proof of Proposition \ref{pro:weighted-KB2} (ii):} It
directly follows Lemma \ref{lem:weighted-KB2}.

Let $\text{conv}(\mathbb{S}(\mathbf{G})\cap \mathcal{NSG}):=\{\mathbf{G}\in%
\mathcal{G}:\exists \alpha\in[0,1], \mathbf{G}_{1}, \mathbf{G}_{2}\in (%
\mathbb{S}(\mathbf{G})\cap \mathcal{NSG}) \text{ s.t. }\mathbf{G}=\alpha%
\mathbf{G}_{1}+(1-\alpha)\mathbf{G}_{2} \}$ be the set of convex
combinations of unweighted networks in $\mathbb{S}(\mathbf{G})\cap \mathcal{%
NSG}$.

\begin{lem}
\label{lem:weighted-KB2} Fix a QC graph $\mathbf{G}$, the following holds,

\begin{enumerate}
\item[(i).] $\arg \max_{\bar{\mathbf{G}} \in \mathbb{S}_w(\mathbf{(}G))} b
(2,\bar{\mathbf{G}}) \subseteq \text{conv} (\mathbb{S}(\mathbf{G})\cap 
\mathcal{NSG})$;

\item[(ii).] For any $k\geq 2$, $\mathbf{1}^\prime \mathbf{QC}^k \mathbf{1}
> \mathbf{1}^\prime \hat{\mathbf{G}}^k \mathbf{1}$, for any $\hat{\mathbf{G}}
\in \text{conv}(\mathbb{S}(\mathbf{G})\cap \mathcal{NSG})\backslash\{\mathbf{%
QC}\}$.
\end{enumerate}
\end{lem}

The proof of Lemma \ref{lem:weighted-KB2} (ii) is presented alongside the
proof of Lemma \ref{lem:unweighted-QC}. Below, we provide only the proof of
Lemma \ref{lem:weighted-KB2} (i).

\noindent\textbf{Proof of Lemma \ref{lem:weighted-KB2} (i).} When $t=1$, the unique weighted NSG in $\mathbb{S}_{w}(\mathbf{G})$ consistent with Proposition \ref{pro:weighted-KB2} (i) is a
QC graph whose clique has size $2$.

Now consider $t\geq 2$. The QC graph $\mathbf{G}$ contains a maximal
complete subgraph on $p\geq 2$ nodes. Concretely, the first $p$ nodes form a
clique; node $p+1$ connects to the first $t-\bar{p}$ nodes, where $\bar{p}:=%
\frac{p(p-1)}{2}$ is the number of links within the clique; and the
remaining $n-(p+1)$ nodes (if any) are isolated. If $t\in\{\bar{p}, \frac{%
p(p+1)}{2}\}$, the set of weighted NSGs in $\mathbb{S}_{w}(\mathbf{G})$ is
obtained by placing additional weight between node $1$ and isolated nodes.
The maximizer is (isomorphic to) an unweighted QC graph; otherwise, node $1$ would
hold two distinct links with weights in $(0,1)$, contradicting Proposition %
\ref{pro:weighted-KB2} (i). Hence, it suffices to analyze the case 
\begin{equation*}
\bar{p}<t<\frac{p(p+1)}{2}.
\end{equation*}

Partition the nodes of $\mathbf{G}$ into four classes according to degrees from (weakly) high to low:
\begin{itemize}
\item Class 1: nodes $1,\dots, t-\bar{p}$;
\item Class 2: nodes $t-\bar{p}+1,\dots, p$;
\item Class 3: node $p+1$;
\item Class 4: nodes $p+2,\dots, n$ (isolated).
\end{itemize}

Let $\mathbf{G}^{\ast}\in\arg\max_{\bar{\mathbf{G}}\in\mathbb{S}_{w}(\mathbf{G})} b(2,\bar{\mathbf{G}})$. By Proposition \ref{pro:weighted-KB2} (i), $\mathbf{G}^{\ast}$ is a weighted NSG. We establish two key claims:

\medskip
\noindent\textbf{Claim (a).} In $\mathbf{G}^{\ast}$, $g_{ij}^{\ast}=0$ whenever $i\in[t-\bar{p}+1,n]\setminus\{p+1\}$ (Classes 2 or 4) and $j\in[p+2,n]$ (Class 4).

\noindent\textit{Proof of Claim (a).} Suppose $g_{ij}^{\ast}>0$ for some such pair. If $g_{1j}^{\ast}=0$, then
nestedness is violated because $g_{1,p+1}^{\ast}=1>g_{i,p+1}^{\ast}$ while $%
g_{1j}^{\ast}=0<g_{ij}^{\ast}$. Hence $g_{1j}^{\ast}>0$. But then node $j$ would
have two positive weighted links, $g_{1j}^{\ast}$ and $g_{ij}^{\ast}$, which
contradicts Proposition \ref{pro:weighted-KB2} (i) since no node can have two
links with weights in $(0,1)$ at the optimum. Thus there are no weights
between Classes 2 and 4, nor within Class 4.

\medskip
\noindent\textbf{Claim (b).} In $\mathbf{G}^{\ast}$, $g_{j,p+1}^{\ast}=0$ for all $j\in[p+2,n]$ (Class 4).

\noindent\textit{Proof of Claim (b).} Otherwise, to preserve nestedness between node $p+1$ and any
node in the clique, the nodes in the clique must connect with $j$, and thus
violate Proposition \ref{pro:weighted-KB2} (i) since $j$ would have multiple links
with weights in $(0,1)$.

\medskip
From the Claims, any positive weights in $\mathbf{G}^{\ast}$ beyond those in $\mathbf{%
G}$ must be placed either on (i) pairs $(i,p+1)$ with $i\in[t-\bar{p}+1,p]$
(Class 2 with Class 3), or (ii) pairs $(i,j)$ with $i\in[1,t-\bar{p}]$
(Class 1) and $j\in[p+2,n]$ (Class 4).

Moreover, by Proposition \ref{pro:weighted-KB2} (i), there is at most one $%
i\in[t-\bar{p}+1,p]$ such that $g_{i,p+1}^{\ast}>0$; otherwise, node $p+1$
would sustain two distinct links with weights in $(0,1)$. Similarly, there
is at most one pair $(i,j)$ with $i\in[1,t-\bar{p}]$, $j\in[p+2,n]$ and $%
g_{ij}^{\ast}>0$; otherwise, preserving nestedness across two isolated nodes
would force a second positive link from some node, again contradicting
Proposition \ref{pro:weighted-KB2} (i).

Therefore, $\mathbf{G}^{\ast}$ differs from $\mathbf{G}$ by assigning
weight to at most two edges: one of type $(i,p+1)$ with $i\in[t-\bar{p}+1,p]$, and one of type $(i,j)$ with $i\in[1,t-\bar{p}]$, $j\in[p+2,n]$. Each such weighted configuration lies in $\mathrm{conv}\big(\mathbb{S}(\mathbf{G})\cap\mathcal{NSG}\big)$ because the corresponding
unweighted augmentations $\mathbf{G}+\mathbf{E}_{ij}$ are NSGs, and any
single-edge weighting is a convex combination of two such unweighted NSGs.
Hence 
\begin{equation*}
\arg\max_{\bar{\mathbf{G}}\in\mathbb{S}_{w}(\mathbf{G})}b(2,\bar{\mathbf{G}})\subseteq\mathrm{conv}\big(\mathbb{S}(\mathbf{G})\cap\mathcal{NSG}\big).
\end{equation*}%
This completes the proof of Lemma \ref{lem:weighted-KB2} (i). \hfill$\Box$

\bigskip

\newpage 
\bibliographystyle{chicagoa}
\bibliography{seq_net_design}

\bigskip

\bigskip

\bigskip

\end{document}